\documentclass{article}
\usepackage[a4paper, total={6in, 8in}]{geometry}
\usepackage{authblk}

\usepackage{times}
\usepackage{soul}
\usepackage{url}
\usepackage[hidelinks]{hyperref}
\usepackage[utf8]{inputenc}
\usepackage[small]{caption}
\usepackage{graphicx}
\usepackage{amsmath}
\usepackage{amsthm}
\usepackage{booktabs, colortbl}
\usepackage{algorithm}
\usepackage{algorithmic}
\usepackage{natbib}

\usepackage{wrapfig}

\usepackage{tcolorbox}
\newtcbox{citBadgeInline}[1][]{ on line,
 enhanced,colframe=blue!50!black,
 colback=yellow!10!white,
 fontupper=\scriptsize,
 baseline = 1.2mm,
 left = -1pt, right = -1pt, top = -2pt, bottom = -2pt,
 frame code={ \path[tcb fill frame]
 ([xshift=-4mm]frame.west)-- (frame.north west) -- (frame.north east) --
 ([xshift=4mm]frame.east)-- (frame.south east) -- (frame.south west) -- cycle;
 }
 ,interior code={ \path[tcb fill interior] ([xshift=-2mm]interior.west)--
 (interior.north west) -- (interior.north east)-- ([xshift=2mm]interior.east) --
 (interior.south east) -- (interior.south west)-- cycle;} , #1
}

\usepackage[T1]{fontenc}
\usepackage{xcolor}
\usepackage{amssymb,amsmath,amsthm}
\usepackage{etoolbox}
\usepackage{booktabs}
\usepackage{tabularx}

\usepackage{paralist}
\usepackage{tikz}
\usetikzlibrary{calc}
\usepackage{mathtools}
\usepackage{doi}
\usepackage[text size=footnotesize]{todonotes}
\usepackage{xargs} 
\newcommandx{\set}[2][1=1]{\ensuremath{\{#1,\ldots,#2\}}}
\newcommandx{\tlog}[3][1=,3=]{\log_{#1}^{#3}(#2)}
 \newtheorem{proposition}{Proposition}
 \newtheorem{corollary}{Corollary}
 \newtheorem{lemma}{Lemma}
 \newtheorem{claim}{Claim}
 \newtheorem{theorem}{Theorem}
  \newtheorem{observation}{Observation}
  
  \newtheorem{rrule}{Reduction Rule}
  
  \theoremstyle{definition}
  \newtheorem{definition}{Definition}
  
  \newtheorem{construction}{Construction}
 \usepackage[sort&compress,nameinlink,noabbrev,capitalize]{cleveref} 

\crefname{observation}{Observation}{Observations}
\crefname{claim}{Claim}{Claims}
\crefname{rrule}{Reduction Rule}{Reduction Rules}
\Crefname{rrule}{Reduction Rule}{Reduction Rules}
\crefname{construction}{Construction}{Constructions}
\Crefname{proposition}{Pr.}{Props.}
\crefname{proposition}{Proposition}{Propositions}
\Crefname{theorem}{Th.}{Th.}
\crefname{theorem}{Theorem}{Theorems}
\Crefname{corollary}{Co.}{Cors.}
\crefname{corollary}{Corollary}{Corollaries}
\crefname{conjecture}{Conjecture}{Conjectures}

\newcommand{\yes}{\texttt{yes}}
\newcommand{\no}{\texttt{no}}
\newcommand{\RD}{$(\Rightarrow)\quad$}
\newcommand{\LD}{$(\Leftarrow)\quad$}

\newcommand{\multicoloredClique}{\textsc{Multi\-colored Clique}}

\newcommand{\np}{\ensuremath{\mathrm{NP}}}

\newcommand{\xp}{\ensuremath{\mathrm{XP}}}                                       
\newcommand{\wone}{\ensuremath{\mathrm{W[1]}}}                                   
\newcommand{\wonehard}{\wone-hard}

\newcommand{\nphard}{\np-hard}

\DeclareMathOperator{\score}{score}                                              
\DeclareMathOperator{\switch}{switch}

\newcommand{\lqed}{}

\newcommandx{\problemdef}[5][2=Input,4=Question]{%
  \begingroup
  \par\medskip
  \noindent \textsc{#1}\nopagebreak[4]
  \par\noindent\hangindent=\parindent\textbf{#2}:  #3\nopagebreak[4]
  \par\noindent\hangindent=\parindent\textbf{#4}:  #5
  \par  \medskip
  \endgroup
}

\newcommand{\N}{\mathbb{N}}
\newcommand{\Q}{\mathbb{Q}}
\newcommand{\Z}{\mathbb{Z}}
\newcommand{\norm}[2]{\ensuremath{\|#1\|_{#2}}}
\renewcommand{\O}{\mathcal{O}}

\newcommand{\prob}[1]{\textnormal{\textsc{#1}}}
\newcommand{\MPV}{\prob{Multistage SNTV}}
\newcommand{\RMPV}{\prob{Revolutionary Multistage SNTV}}
\newcommand{\mpv}{\prob{MSNTV}}
\newcommand{\cmpv}{\prob{MSNTV}}%
\newcommand{\rmpv}{\prob{RMSNTV}}
\newcommand{\VC}{\prob{Vertex Cover}}

\newcommand{\cocl}[1]{\ensuremath{\operatorname{#1}}}
\newcommand{\W}[1]{\cocl{W[#1]}}
\newcommand{\NP}{\cocl{NP}}
\newcommand{\FPT}{\cocl{FPT}}

\newcommand{\coNP}{\cocl{coNP}}
\newcommand{\XP}{\cocl{XP}}
\newcommand{\poly}{\cocl{poly}}
\newcommand{\sign}{\cocl{sign}}

\newcommand{\NPincoNPslashpoly}{\ensuremath{\NP\subseteq \coNP/\poly}}

\newcommand{\calC}{\mathcal{C}}

\newcommand{\symdif}[2]{\ensuremath{#1\triangle#2}}
\newcommand{\cqed}{\hfill$\diamond$}
\newcommand{\wmod}[1]{\ensuremath{\widehat{#1}}}
\newcommand{\struct}[1]{\noindent\textbf{#1}}
\newcommand{\tref}[1]{\textsuperscript{\tiny(\Cref{#1})}}
\newcommand{\trefs}[2]{\scriptsize(\Cref{#1,#2})}
\newcommand{\apptref}[1]{\textsuperscript{\tiny(\appsymb)}}

\newcommand{\ceq}{\ensuremath{\coloneqq}}

\definecolor{lilla}{HTML}{750787}
\newcommand{\swColA}{cyan}
\newcommand{\swColB}{lilla}

\newcommand{\xcase}[2]{\smallskip\noindent\textbf{Case~#1: #2.}}
\newcommand{\xsubcase}[2]{\smallskip\noindent\hspace{1em}\textbf{Case~#1: #2.}}
\newcommand{\xsubsubcase}[3]{\smallskip\noindent\hspace{2em}\textbf{Case~#1.#2:
#3.}}

\newcommand{\appsymb}{$\bigstar$}
\newcommand{\appref}[1]{{\hyperref[proof:#1]{\appsymb}}}
\newcommand{\apprefX}[1]{{\hyperref[#1]{\appsymb}}}

\newcommand{\appendixsection}[1]{%
  \gappto{\appendixProofText}{\section{Additional Material for Section~\ref{#1}}\label{app:#1}}
}

\newcommand{\toappendix}[1]{%
  \gappto{\appendixProofText}
  {{
    #1
  }}
}
\newcommand{\appendixproof}[2]{%
  \gappto{\appendixProofText}
  {
    \subsection{Proof of \cref{#1}}\label{proof:#1}
    #2
  }
}

\newcommand{\thextitle}{When Votes Change and Committees Should (Not)}

\title{\thextitle}

\author[1,2]{Robert~Bredereck}
\author[1,3]{Till~Fluschnik\thanks{Supported by DFG, project AFFA~(BR~5207/1)}}
\author[3,4]{Andrzej~Kaczmarczyk}
\affil[1]{{\small Institut für Informatik, TU Clausthal, Clausthal, Germany}}
\affil[2]{{\small Humboldt-Universit\"at zu Berlin, Berlin, Germany}}
\affil[3]{{\small Technische Universität Berlin, Faculty~IV, Algorithmics and
Computational Complexity, Germany}}
\affil[4]{{\small AGH University, Krak\'ow, Poland}}
\affil[ ]{{\small \textsl{
robert.bredereck@tu-clausthal.de,
till.fluschnik@tu-clausthal.de,
andrzej.kaczmarczyk@agh.edu.pl
}}}

\usepackage{cleveref}

\newcommand{\Wcmpv}{\textsc{W-\cmpv{}}}

\begin{document}

\maketitle 

\begin{abstract}
Electing a single committee of a small size is a classical and well-understood voting
situation. %
Being interested in a sequence of committees, we introduce %
two time-dependent
multistage models based on simple scoring-based voting.
Therein,
we are given a sequence of voting profiles (stages) over the same set of agents and candidates,
and our task is to find a small committee for each stage of high score.
In the \emph{conservative model} we additionally require that any two consecutive
committees have a small %
symmetric difference.
Analogously, in the \emph{revolutionary model} we require large %
symmetric differences.
We prove both models to be~\NP-hard even for a constant number of agents,
and, based on this,
initiate a parameterized complexity analysis for the most natural parameters
and combinations thereof. %
Among other results,
we prove both models to be in \XP{} yet \W{1}-hard regarding the number of stages,
and that being revolutionary seems to be ``easier'' than being conservative.

\end{abstract}

\section{Introduction}
\label{sec:intro}
\appendixsection{sec:intro}

In well-studied classical committee election scenarios, given a set of
candidates, we aim at selecting a small committee that is, in a certain sense,
most suitable for a given collection of preferences over the
candidates~\citep{ComSoc2016,TrendsComSoc2017,EandComp2016}.
However, typically these scenarios concentrate solely on electing a
committee in a single election to the neglect of a time dimension. 
This neglect results in serious limitations of the model.

For instance,
it is not
possible to ensure a relationship 
(e.g.,\ a small number of changes)
between any two consecutive committees.
We tackle 
this issue
by introducing a \emph{multistage}~\citep{GuptaTW14}
variant of the problem. In this variant, a \emph{sequence} of voting profiles
\emph{is given}, and we seek a \emph{sequence} of small committees, each
collecting a reasonable number of approvals, 
such that 
the \emph{difference
between consecutive committees} is 
upper-bounded.

For instance,
assume a research community to seek organizers of a series of events (say those scheduled for next year).
The organizers must be fixed in advance to allow %
preparation time.
Since events 
may
differ in location and type, not every candidate fits equally well to every event.
Thus, 
each member (agent) of the community is asked to name one 
suitable
organizer for each event
and, based on this, the goal is to determine a sequence of organizing committees (one for each event).
Naturally, there are three constraints:
(i) each committee is bounded in size,
(ii) each committee has enough support from the agents, and
(iii) at least a certain number of candidates in consecutive committees overlap
to avoid a lack of knowledge transfer jeopardizing effectiveness.

Initiating a study of so far overlooked multistage
variant of multiwinner elections, we aim at understanding the computational
complexity of 
the related computational problems. In particular we want
to detect computationally tractable cases. Notably, our
multistage setting introduces two new dimensions to the standard model of
multiwinner elections: the relation between consecutive committees and the time.
Thus, our second goal is to observe how these two dimensions affect the
computational complexity of the introduced model.

\begin{figure*}[t]
  \centering
  \begin{tikzpicture}
  
      \usetikzlibrary{patterns,backgrounds,shapes,calc,shapes.arrows,shapes.misc}
      \def\xr{.85}
      \def\yr{1}

      \def\boxw{4}
      \def\boxh{1.075}

      \tikzstyle{xarc}=[->,gray,thick,>=latex]
      \tikzstyle{indicDot} = [fill=black, text width = 1.5pt, draw =
                      black, circle, inner sep = 0pt]
      \tikzstyle{indicLine} = [dotted]

      \newcommand{\parabox}[9]{
      \node (a#1) at (#2)[rectangle, draw, rounded corners, minimum width=\xr*\boxw cm, minimum height=\yr*\boxh cm,fill=gray,draw,thick]{};

      \node (top#1) at (a#1.north)[anchor=north,rectangle, rounded corners, minimum
      width=\xr*\boxw cm, minimum height=\yr*\boxh*0.25
      cm,fill=lightgray!10!white,thick,draw, text depth = 3pt, inner ysep = 1pt]{{\phantom{$\ell$}#3\phantom{$\ell$}}};

      \node at (top#1.south west)[anchor=north west,rectangle, rounded corners,
      minimum width=\xr*\boxw*0.5 cm, minimum height=\yr*\boxh*0.575
      cm,fill=#6,align=center,font=\footnotesize,draw, inner ysep = 1pt, text
      height = 13pt]{\vphantom{$l^l$}#4#8};

      \node at (top#1.south east)[anchor=north east,rectangle, rounded corners,
      minimum width=\xr*\boxw*0.5 cm, minimum height=\yr*\boxh*0.575
      cm,fill=#7,align=center,font=\footnotesize,draw, inner ysep = 1pt, text
      height = 13pt]{\vphantom{$l^l$}#5#9};

      \node (leftProb#1) at (top#1.south west)[anchor=west,
      text width=\xr*\boxw*0.275 cm, inner ysep = -1pt,
      shape=chamfered rectangle, chamfered rectangle corners={north east, south east},
      chamfered rectangle angle=70,
      chamfered rectangle xsep=3em,draw,fill=white]{\vphantom{\cmpv}};

      \node (rightProb#1) at (top#1.south east)[anchor=east,
      text width=\xr*\boxw*0.275 cm, inner ysep = -1pt,
      shape=chamfered rectangle, chamfered rectangle corners={north west,
      south west},
      chamfered rectangle angle=70,
      chamfered rectangle xsep=3em,draw,fill=white]{\vphantom{\rmpv{}}};
      
      \node at (leftProb#1.west)[anchor=west,draw=none,fill=none, inner sep=1pt]{\scriptsize \cmpv{}};
      \node at (rightProb#1.east)[anchor=east,draw=none,fill=none, inner
      sep=1pt]{\scriptsize
      \rmpv{}};
      }

      \def\xsh{4.5}
      \def\ysh{1.8}

      \parabox{m}{-1.5*\xsh*\xr,0*\yr}{$m$}{FPT}{FPT}{green!20!white}{green!20!white}{\tref{cor:FPTwrtm}}{\tref{cor:FPTwrtm}};
      \parabox{k}{-1.5*\xsh*\xr,-\ysh*\yr}{\small$k$}{XP, W[2]-h.}{FPT$^\dagger$}{orange!20}{green!20!white}{\textsuperscript{a}}{};
      \parabox{ell}{-1.5*\xsh*\xr,-2*\ysh*\yr}{\small$\ell$}{p-NP-h}{FPT$^\dagger$}{red!18!white}{green!20!white}
      {\tref{thm:bothnphard}}{};
      \draw[xarc] (aell) to (ak);
      \draw[xarc] (ak) to (am);

      \parabox{n}{1.5*\xsh*\xr,0*\yr}{$n$}{p-NP-h}{p-NP-h}{red!18!white}{red!18!white}{\tref{thm:bothnphard}}{\tref{thm:bothnphard}};
      \parabox{x}{1.5*\xsh*\xr,-\ysh*\yr}{$x$}{p-NP-h}{p-NP-h}{red!18!white}{red!18!white}
      {\tref{thm:bothnphard}}{\tref{thm:bothnphard}};
      \draw[xarc] (ax) to (an);

      \parabox{tau}{-0.5*\xsh*\xr,-2*\ysh*\yr}{$\tau$}{XP, W[1]-h.}{XP, W[1]-h}{orange!20}{orange!20}{\textsuperscript{b}}{\textsuperscript{c}};
      
      \parabox{elltau}{-0.5*\xsh*\xr,-1*\ysh*\yr}{$\ell+\tau$}{XP, W[1]-h.}{FPT$^\dagger$}{orange!20}{green!20!white}{\textsuperscript{b}}{};
      \parabox{ktau}{-0.5*\xsh*\xr,0*\ysh*\yr}{$k+\tau$}{XP, W[1]-h.}{FPT$^\dagger$}{orange!20}{green!20!white}{\textsuperscript{b}}{};
      \parabox{xtau}{0.5*\xsh*\xr,0*\ysh*\yr}{$x+\tau$}{FPT\apptref{prop:cmpvxtau}}{XP or FPT?}{green!20}{yellow!20}{}{};

      \draw[xarc] (ak) to (aktau);
      \draw[xarc] (aell) to (aelltau);
      \draw[xarc] (ax) to (axtau);
      
      \parabox{mn}{0*\xsh*\xr,\ysh*\yr}{\normalsize$m+n$}{FPT}{FPT}{green!20}{green!20}{
      }{
      };
      \parabox{mtau}{-1.5*\xsh*\xr,\ysh*\yr}{\normalsize$m+\tau$}{FPT}{FPT}{green!20!white}{green!20!white}{\tref{thm:PKwrtmtau}}{\tref{thm:PKwrtmtau}};
      \parabox{ntau}{1.5*\xsh*\xr,\ysh*\yr}{\normalsize$n+\tau$}{FPT}{FPT}{green!20!white}{green!20!white}{\tref{thm:pkntau}}{\tref{thm:pkntau}};

      \draw[xarc] (am.north east) to (amn);
      \draw[xarc] (am) to (amtau);
      \draw[xarc] (an.north west) to (amn);
      \draw[xarc] (an) to (antau);
      \draw[xarc] (atau) to (aelltau);
      \draw[xarc] (aelltau) to (aktau);
      \draw[xarc] (aktau) to (amtau);
      \draw[xarc] ($(atau.north east)-(0.2*\xr,0)$) to ($(axtau.south west)+(0.3*\xr,0)$);
      \draw[xarc] (axtau) to (antau);
      
      \begin{scope}[xshift=\xr*2.375 cm,yshift=\yr*-2.5cm,scale=0.8]
        \newcommand{\swD}{14} %
        \newcommand{\swDm}{13} %
        \newcommand{\swU}{4} %

        \newdimen\swR %
        \swR=1.95cm 
        \newdimen\swL %
        \swL=2.15cm

        \newcommand{\swA}{360/\swD} %
        
        \newcommand{\swFontRing}{\scriptsize\color{darkgray}}
        \draw[fill=green!15!white,draw=none] circle [radius=\swR*4/4];
        \draw[fill=yellow!15!white,draw=none] circle [radius=\swR*3/4];
        \draw[fill=orange!15!white,draw=none] circle [radius=\swR*2/4];
        \draw[fill=red!15!white,draw=none] circle [radius=\swR*1/4];

        \def \hookx {1em}
        \def \hooky {-1.5ex}
        \path (270:1*\swR/\swU-0.6*\swR/\swU) node[xshift = 9em, yshift =
        \hooky, anchor = west] (pnpLab) {\swFontRing p-NP-h};
        \path[dotted] (270:2*\swR/\swU-0.4*\swR/\swU) node[xshift = 9em, yshift =
        \hooky, anchor = west] (xpwLab) {\swFontRing  XP, W-h};
        \path (270:3*\swR/\swU-0.6*\swR/\swU) node[xshift = 9em, yshift =
        \hooky, anchor = west] (xpLab) {\swFontRing  XP};
        \path (270:4*\swR/\swU-0.6*\swR/\swU) node[xshift = 9em, yshift =
        \hooky, anchor = west] (fptLab) {\swFontRing  FPT};

        \node[indicDot] (pnpDot) at (250:1*\swR/\swU-0.25*\swR/\swU) {};
        \node[indicDot] (xpwDot) at (250:2*\swR/\swU-0.25*\swR/\swU) {};
        \node[indicDot] (xpDot) at (250:3*\swR/\swU-0.25*\swR/\swU) {};
        \node[indicDot] (fptDot) at (250:4*\swR/\swU-0.25*\swR/\swU) {};

        \coordinate (pnpTurn) at ($(pnpDot)  + (\hookx, 0)$);
        \coordinate (xpwTurn) at ($(xpwDot)  + (\hookx, 0)$);
        \coordinate (xpTurn) at ($(xpDot)  + (\hookx, 0)$);
        \coordinate (fptTurn) at ($(fptDot)  + (\hookx, 0)$);

        \draw[indicLine]  (pnpLab.west) --  (pnpTurn |- pnpLab.west) -- (pnpDot);
        \draw[indicLine]  (xpwLab.west) --  (xpwTurn |- xpwLab.west) -- (xpwDot);
        \draw[indicLine]  (xpLab.west)  --  (xpTurn |- xpLab.west)   -- (xpDot); 
        \draw[indicLine]  (fptLab.west) --  (fptTurn |- fptLab.west) -- (fptDot);
        \path (0:0cm) coordinate (O); %
          \foreach \swX in {1,...,\swD}{
            \draw[thin,color=lightgray] (\swX*\swA:0) -- (\swX*\swA:\swR);
          }

          \foreach \swY in {1,...,\swU}{
            \foreach \swX in {1,...,\swD}{
              \path (\swX*\swA:\swY*\swR/\swU-0.5*\swR/\swU) coordinate (D\swX-\swY);
            }
          }
          
          \newcommand{\swPolygon}[3]{
            \foreach \x/\y in {#1}{\node (t\x) at (D\x-\y)[scale=0.4,draw,#2]{};}
              \foreach[evaluate=\x as \y using int(\x+1)] \x in {1,2,...,\swDm}{\draw[#3] (t\x) to (t\y);}\draw[#3] (t\swD) to (t1);
          }
          \foreach \x/\y in {
          1/$x$,
          2/$k$,
          3/$m$,
          4/$n$,
          5/$m+n$\phantom{AA},
          6/$\tau$,
          7/$k+\tau$~~~~,
          8/$m+\tau$,
          9/\phantom{AA}$x+\tau$,
          10/\phantom{AA}$n+\tau$,
          11/$\ell$,
          12/$\ell+n$,
          13/\phantom{AA}$k+x$,
          14/\phantom{AA}$k+n$
          }{\path (\x*\swA:\swL) node (L\x) {\scriptsize \y};}

        \swPolygon{1/1,2/2,3/4,4/1,5/4,6/2,7/2,8/4,9/4,10/4,11/1,12/1,13/2,14/3}{
        circle,
        draw=\swColA!50!black,fill=\swColA,scale=0.85}{thick,\swColA}
        \swPolygon{1/1,2/4,3/4,4/1,5/4,6/2,7/4,8/4,9/3,10/4,11/4,12/4,13/4,14/4}{fill=\swColB,draw=\swColB!50!black,scale=0.66}{thick,densely dashed,\swColB}
      \end{scope}

    \end{tikzpicture}
    \caption{
    Overview of results for \cmpv{} and \rmpv{}. Abbreviations
    p-\NP-h and~\W{1}-h stand for, respectively,
    para-\NP-hard and~\W{1}-hard. An arrow from one
    parameter~$p$ to another parameter~$p'$ indicates that~$p$ can be upper
    bounded by some function in~$p'$ (e.g.,~$\ell\leq 2k$, $k\leq m$, or~$x\leq
    n$). 
    The spiderweb diagram depicts further results being not displayed for readability
    (solid: conservative; dashed: revolutionary).
    ~~\textsuperscript{a}\,\trefs{thm:xpregardingk}{thm:cmpvwhardktau}
    ~~\textsuperscript{b}\,\trefs{thm:xptau}{thm:cmpvwhardktau}
    ~~\textsuperscript{c}\,\trefs{thm:xptau}{thm:rmpvwhardtau}
    ~~$^\dagger$\,\protect\citep{KRZ21}
    }
    \label{fig:results}
  \end{figure*}

\subsection{Model and Examples}

We denote by~$\N$ and~$\N_0$ the natural numbers excluding and including zero, respectively.
For some function~$f\colon A\to B$,
let~$f^{-1}(B')=\{a\in A\mid f(a)\in B'\}$ for every~$B'\subseteq B$.
We use basic notation from graph theory~\citep{Diestel10}
and parameterized algorithmics~\citep{cygan2015parameterized}.
The main problem of this work is as follows.
\problemdef{\MPV{} (\mpv)}
{A set of agents~$A=\{a_1,\ldots,a_n\}$, 
a set of candidates~$C=\{c_1,\ldots,c_m\}$, 
a sequence~$U=(u_1,\ldots,u_\tau)$ of~$\tau$ voting profiles with~$u_t\colon
A \to C \cup \{\emptyset\}$, $t \in \set{\tau}$, 
and three integers~$k\in\N,$ $\ell\in\N_0$, and~$x\in\N$.}
{Is there a sequence~$(C_1,\ldots,C_\tau)$ of committees~$C_t\subseteq C$ 
such that for all~$t\in \set{\tau}$ it holds true that
$|C_t|\leq k$ and
$\score_t(C_t)\ceq |u_t^{-1}(C_t)|\geq x$,
and 
\begin{align}
  |\symdif{C_t}{C_{t+1}}|&\leq \ell   \label{eq:mpv}
\end{align}
holds true for all~$t\in\set{\tau-1}$?
}
  
\noindent
One may wonder why we chose an upper bound on~$k$ in the problem definition
instead of specifying an exact constant committee size.
While most natural instances will have solutions with committees of size exactly~$k$,
requiring them rules out odd symmetric differences, e.g., $\ell=1$.

\toappendix{
 \centering
 \begin{figure*}[t]
  \begin{tikzpicture}
  
      \usetikzlibrary{patterns,backgrounds,shapes,calc}
      \def\xr{.85}
      \def\yr{1}

      \def\boxw{4}
      \def\boxh{1.4}

      \tikzstyle{xarc}=[->,gray,thick,>=latex]
      \tikzstyle{indicDot} = [fill=black, text width = 1.5pt, draw =
                      black, circle, inner sep = 0pt]
      \tikzstyle{indicLine} = [dotted]

      \newcommand{\parabox}[9]{
      \node (a#1) at (#2)[rectangle, rounded corners, minimum width=\xr*\boxw cm, minimum height=\yr*\boxh cm,fill=gray,draw,very thick]{};
      
      \node at (a#1.south west)[anchor=south west,rectangle, rounded corners, minimum width=\xr*\boxw*0.5 cm, minimum height=\yr*\boxh*0.7 cm,draw,fill=white,label={[yshift=-\yr*\boxh*0.375 cm]90:{\footnotesize \cmpv{}}}]{};
      \node at (a#1.south west)[anchor=south west,rectangle, rounded corners,
      minimum width=\xr*\boxw*0.5 cm, minimum height=\yr*\boxh*0.25
      cm,fill=#6,align=center,font=\scriptsize,draw]{\vphantom{$l^l$}#4#8};
      \node at (a#1.south east)[anchor=south east,rectangle, rounded corners, minimum width=\xr*\boxw*0.5 cm, minimum height=\yr*\boxh*0.7 cm,fill=white,draw,label={[yshift=-\yr*\boxh*0.375 cm]90:{\footnotesize \rmpv{}}}]{};
      \node at (a#1.south east)[anchor=south east,rectangle, rounded corners,
      minimum width=\xr*\boxw*0.5 cm, minimum height=\yr*\boxh*0.25
      cm,fill=#7,align=center,font=\scriptsize,draw]{\vphantom{$l^l$}#5#9};
      \node at (a#1.north)[anchor=north,rectangle, rounded corners, minimum width=\xr*\boxw cm, minimum height=\yr*\boxh*0.25 cm,fill=lightgray!10!white,thick,draw]{{\phantom{$\ell$}#3\phantom{$\ell$}}};
      }

      \def\xsh{4.5}
      \def\ysh{1.8}

      \parabox{m}{-1.5*\xsh*\xr,0*\yr}{$m$}{FPT, noPK}{FPT, noPK}{green!20!white}{green!20!white}{\tref{cor:FPTwrtm}}{\tref{cor:FPTwrtm}};
      \parabox{k}{-1.5*\xsh*\xr,-\ysh*\yr}{\small$k$}{XP, W[2]-h.}{FPT$^\dagger$, noPK}{orange!20}{green!20!white}{\textsuperscript{a}}{};
      \parabox{ell}{-1.5*\xsh*\xr,-2*\ysh*\yr}{\small$\ell$}{p-NP-h}{FPT$^\dagger$, noPK}{red!18!white}{green!20!white}{\tref{prop:cmpvnphard}}{};
      \draw[xarc] (aell) to (ak);
      \draw[xarc] (ak) to (am);

      \parabox{n}{1.5*\xsh*\xr,0*\yr}{$n$}{p-NP-h}{p-NP-h}{red!18!white}{red!18!white}{\tref{thm:bothnphard}}{\tref{thm:bothnphard}};
      \parabox{x}{1.5*\xsh*\xr,-\ysh*\yr}{$x$}{p-NP-h}{p-NP-h}{red!18!white}{red!18!white}{\tref{prop:cmpvnphard}}{\tref{cor:rmpvnphard}};
      \draw[xarc] (ax) to (an);

      \parabox{tau}{-0.5*\xsh*\xr,-2*\ysh*\yr}{$\tau$}{XP, W[1]-h.}{XP, W[1]-h}{orange!20}{orange!20}{\textsuperscript{b}}{\textsuperscript{c}};
      
      \parabox{elltau}{-0.5*\xsh*\xr,-1*\ysh*\yr}{$\ell+\tau$}{XP, W[1]-h.}{FPT$^\dagger$, PK?}{orange!20}{green!20!white}{\textsuperscript{b}}{};
      \parabox{ktau}{-0.5*\xsh*\xr,0*\ysh*\yr}{$k+\tau$}{XP, W[1]-h.}{FPT$^\dagger$, PK?}{orange!20}{green!20!white}{\textsuperscript{b}}{};
      \parabox{xtau}{0.5*\xsh*\xr,0*\ysh*\yr}{$x+\tau$}{FPT\tref{prop:cmpvxtau}, PK?}{XP, FPT?}{green!20}{yellow!20}{}{};

      \draw[xarc] (ak) to (aktau);
      \draw[xarc] (aell) to (aelltau);
      \draw[xarc] (ax) to (axtau);
      
      \parabox{mn}{0*\xsh*\xr,\ysh*\yr}{\normalsize$m+n$}{FPT, noPK}{FPT, noPK}{green!20}{green!20}{
      }{\tref{thm:nopkmn}};
      \parabox{mtau}{-1.5*\xsh*\xr,\ysh*\yr}{\normalsize$m+\tau$}{FPT, PK}{FPT, PK}{green!10!white}{green!10!white}{\tref{thm:PKwrtmtau}}{\tref{thm:PKwrtmtau}};
      \parabox{ntau}{1.5*\xsh*\xr,\ysh*\yr}{\normalsize$n+\tau$}{FPT, PK}{FPT, PK}{green!10!white}{green!10!white}{\tref{thm:pkntau}}{\tref{thm:pkntau}};

      \draw[xarc] (am.north east) to (amn);
      \draw[xarc] (am) to (amtau);
      \draw[xarc] (an.north west) to (amn);
      \draw[xarc] (an) to (antau);
      \draw[xarc] (atau) to (aelltau);
      \draw[xarc] (aelltau) to (aktau);
      \draw[xarc] (aktau) to (amtau);
      \draw[xarc] (atau.north east) to ($(axtau.south west)+(0.3*\xr,0)$);
      \draw[xarc] (axtau) to (antau);
      
      \begin{scope}[xshift=\xr*2.375 cm,yshift=\yr*-2.9cm,scale=0.8]
        \newcommand{\swD}{14} %
        \newcommand{\swDm}{13} %
        \newcommand{\swU}{5} %

        \newdimen\swR %
        \swR=1.95cm 
        \newdimen\swL %
        \swL=2.15cm

        \newcommand{\swA}{360/\swD} %
        
        \newcommand{\swFontRing}{\scriptsize\color{darkgray}}
        \draw[fill=green!7!white,draw=none] circle [radius=\swR*5/5];
        \draw[fill=green!15!white,draw=none] circle [radius=\swR*4/5];
        \draw[fill=yellow!15!white,draw=none] circle [radius=\swR*3/5];
        \draw[fill=orange!15!white,draw=none] circle [radius=\swR*2/5];
        \draw[fill=red!15!white,draw=none] circle [radius=\swR*1/5];

        \def \hookx {1em}
        \def \hooky {-.75ex}
        \path (270:1*\swR/\swU-0.5*\swR/\swU) node[xshift = 9em, yshift =
        \hooky, anchor = west] (pnpLab) {\swFontRing p-NP-h};
        \path (270:2*\swR/\swU-0.5*\swR/\swU) node[xshift = 9em, yshift =
        \hooky, anchor = west] (xpwLab) {\swFontRing  XP, W-h};
        \path (270:3*\swR/\swU-0.5*\swR/\swU) node[xshift = 9em, yshift =
        \hooky, anchor = west] (xpLab) {\swFontRing  XP};
        \path (270:4*\swR/\swU-0.5*\swR/\swU) node[xshift = 9em, yshift =
        \hooky, anchor = west] (fptLab) {\swFontRing  FPT, noPK};
        \path (270:5*\swR/\swU-0.5*\swR/\swU) node[xshift = 9em, yshift =
        \hooky, anchor = west] (fptpkLab) {\swFontRing  FPT, PK};

        \node[indicDot] (pnpDot) at (250:1*\swR/\swU-0.25*\swR/\swU) {};
        \node[indicDot] (xpwDot) at (250:2*\swR/\swU-0.25*\swR/\swU) {};
        \node[indicDot] (xpDot) at (250:3*\swR/\swU-0.25*\swR/\swU) {};
        \node[indicDot] (fptDot) at (250:4*\swR/\swU-0.25*\swR/\swU) {};
        \node[indicDot] (fptpkDot) at (250:5*\swR/\swU-0.25*\swR/\swU) {};

        \coordinate (pnpTurn) at ($(pnpDot)  + (\hookx, 0)$);
        \coordinate (xpwTurn) at ($(xpwDot)  + (\hookx, 0)$);
        \coordinate (xpTurn) at ($(xpDot)  + (\hookx, 0)$);
        \coordinate (fptTurn) at ($(fptDot)  + (\hookx, 0)$);
        \coordinate (fptpkTurn) at ($(fptpkDot)  + (\hookx, 0)$);

        \draw[indicLine]  (pnpLab.west) --  (pnpTurn |- pnpLab.west) -- (pnpDot);
        \draw[indicLine]  (xpwLab.west) --  (xpwTurn |- xpwLab.west) -- (xpwDot);
        \draw[indicLine]  (xpLab.west)  --  (xpTurn |- xpLab.west)   -- (xpDot); 
        \draw[indicLine]  (fptLab.west) --  (fptTurn |- fptLab.west) -- (fptDot);
        \draw[indicLine]  (fptpkLab.west) --  (fptpkTurn |- fptpkLab.west) --
        (fptpkDot);
        \path (0:0cm) coordinate (O); %
          \foreach \swX in {1,...,\swD}{
            \draw[thin,color=lightgray] (\swX*\swA:0) -- (\swX*\swA:\swR);
          }

          \foreach \swY in {1,...,\swU}{
            \foreach \swX in {1,...,\swD}{
              \path (\swX*\swA:\swY*\swR/\swU-0.5*\swR/\swU) coordinate (D\swX-\swY);
            }
          }
          
          \newcommand{\swPolygon}[3]{
            \foreach \x/\y in {#1}{\node (t\x) at (D\x-\y)[scale=0.4,draw,#2]{};}
              \foreach[evaluate=\x as \y using int(\x+1)] \x in {1,2,...,\swDm}{\draw[#3] (t\x) to (t\y);}\draw[#3] (t\swD) to (t1);
          }
          \foreach \x/\y in {
          1/$x$,
          2/$k$,
          3/$m$,
          4/$n$,
          5/$m+n$\phantom{AA},
          6/$\tau$,
          7/$k+\tau$,
          8/$m+\tau$,
          9/\phantom{AA}$x+\tau$,
          10/\phantom{AA}$n+\tau$,
          11/$\ell$,
          12/$\ell+n$,
          13/\phantom{AA}$k+x$,
          14/\phantom{AA}$k+n$
          }{\path (\x*\swA:\swL) node (L\x) {\scriptsize \y};}

        \swPolygon{1/1,2/2,3/4,4/1,5/4,6/2,7/2,8/5,9/4,10/5,11/1,12/1,13/2,14/3}{
        circle,
        draw=\swColA!50!black,fill=\swColA,scale=0.85}{thick,\swColA}
        \swPolygon{1/1,2/4,3/4,4/1,5/4,6/2,7/4,8/5,9/3,10/5,11/4,12/4,13/4,14/4}{fill=\swColB,draw=\swColB!50!black,scale=0.66}{thick,densely dashed,\swColB}
      \end{scope}

    \end{tikzpicture}
    \caption{
    Overview of results for \cmpv{} and \rmpv{}. 
    Abbreviations PK, noPK, p-\NP-h, and~\W{1}-h stand for, respectively,
    ``polynomial kernel'', ``no polynomial kernel
    unless~\NPincoNPslashpoly'', para-\NP-hard, and~\W{1}-hard. An arrow from one
    parameter~$p$ to another parameter~$p'$ indicates that~$p$ can be upper
    bounded by some function in~$p'$ (e.g.,~$\ell\leq 2k$, $k\leq m$, or~$x\leq
    n$). 
    ``?'' indicates that a detailed classification 
    is unknown.
    The spiderweb diagram depicts further results being not displayed for readability
    (solid: conservative; dashed: revolutionary).
    ~~\textsuperscript{a}\,\trefs{thm:xpregardingk}{thm:cmpvwhardktau}
    ~~\textsuperscript{b}\,\trefs{thm:xptau}{thm:cmpvwhardktau}
    ~~\textsuperscript{c}\,\trefs{thm:xptau}{thm:rmpvwhardtau}
    ~~$^\dagger$\,\protect\citep{KRZ21}
    }
    \label{app:fig:results}
  \end{figure*}
}

SNTV comes from ``single non-transferable vote'',
to which our model boils down
for a single stage.
While most of our results transfer to the setting of general approval
profiles (see~\cref{sec:conclusion}),
we use
''Plurality'' profiles,
where each agent approves 
exactly one candidate,
for four reasons.
(I)~Aiming for
positive algorithmic results and for recognizing the influence of the basic
model properties,
it is
most natural to start
with the simplest relevant scenario.
Even though SNTV is simple,
(II)~it is
the only committee scoring rule serving finding representation- \emph{and}
excellence-focused committees~\citep{FSST19}. Hence, it forms a good basis
for a further exploration of our model for another rules
reaching these two goals.
(III)~Plurality profiles are widely accepted 
in practice and 
form
complex voting procedures (e.g.,\,STV, the two-vote
or two-stage voting systems used for the German or French parliament).
(IV)~Selecting %
a single candidate
is only a weak (cognitive) barrier for human agents increasing the applicability
of the model. 
In fact, our definition can be easily
extended to more expressive scoring-based voting profiles.

Motivated by respective requirements in many applications, \citet{BN21} propose
a systematic study of (multiwinner) voting models that handle
multiple preference profiles at once 
(e.g., when incorporating changes over time).
Following their work, our paper provides one of the first models opening the
field for new potential applications.
Indeed, \cmpv{} models various possible practical scenarios,
two of which we briefly sketch below.
{%

\noindent
\paragraph{Buffet Selection.}
Suppose we are asked to organize the venue's breakfast buffet of a multiday event 
(like a workshop seminar).
We offer different disjoint food bundles 
(candidates) 
for breakfast
and ask the participants of the event 
to share their preferences of which 
bundle is their favorite for which day.
Due to limited space,
we can offer only at most some number
of bundles in the buffet (committee) simultaneously.
Moreover,
to stay at low cost and to avoid food waste,
we want that few bundles change from one day to the next.
Clearly,
given these constraints
and the collected preferences (voting profiles),
we want at all days 
to have a high number
of participants whose voted bundle made it into the buffet.

\noindent
\paragraph{Exhibition Composition.}
When planning a multiday exhibition of sculptures (candidates)
in a lobby of a hotel 
where we are enabled 
neither to show at once all the sculptures that we want to exhibit,
nor to exchange arbitrarily many sculptures between consecutive exhibitions days 
(due to,
e.g.,\ limited capacity of transporting sculptures between some depot and the hotel).
To nevertheless offer an enjoyable experience for numerous visitors,
we ask the visitors to vote for each day they plan to visit
for their favorite sculpture to be exhibited (to keep the poll simple and robust).

Note that if we drop condition~\eqref{eq:mpv}
(or, equivalently, set~$\ell=2k$),
then on the one hand,
we have no control over changes between consecutive committees,
yet on the other hand,
we obtain a linear-time solvable problem.
Thus,
control comes with a computational cost,
bound in the value of~$\ell$.
To proceed with our second goal---in particular,
better understanding of condition~\eqref{eq:mpv}, %
we additionally study a problem variant of~\cmpv{}. We obtain this variant,
referred to as \RMPV{}~(\rmpv{}),
by replacing~\eqref{eq:mpv} by~$|\symdif{C_t}{C_{t+1}}|\geq \ell$.
In words,
in~\rmpv{} we request a change of size at least~$\ell$ between consecutive committees.
By this,
while we complemented the meaning of~$\ell$,
it still expresses a control over the changes,
and hence comes possibly again with a computational cost.
We investigate whether 
the ``\emph{conservative}'' (\cmpv{}) and the revolutionary (\rmpv{}) variant differ, 
and if so, 
then how and why.

\subsection{State of the Art and Our Contributions}

Our model follows the recently proposed \emph{multistage}
model~\citep{EMS14,GuptaTW14}, that led to several multistage
problems~\citep{BampisELP18,BampisET19,FluschnikNRZ22,FluschnikNSZ23,ChimaniTW22,HeegerHKNRS21,BampisEST19,BampisEM18,KRZ21,Fluschnik21,FluschnikK22} next to
ours. Graph problems considered in the multistage model often study classic
problems on temporal graphs (a sequence of graphs over the same vertices).
While all the multistage problems known from literature cover the variant we call
``conservative,'' our revolutionary variant forms a novel submodel herein.

Although, to the best of our knowledge our model is novel, some other aspects of
selecting multiple (sub)committees have been studied in (computational)
social choice theory. The closest is %
a recent work of
\citet{BKN20}, who augment classic
multiwinner elections with a time dimension. Accordingly, they consider
selecting a sequence of committees. However, the major differences
with our work are, first, that they do not allow agents (voters) to change their ballots
over time and, second, that there are no explicit constraints on the differences
between two successive committees.
\citet{FZC17}, \citet{Lac20}, and \citet{PP13}
allow this but they consider an online scenario (in contrary to our problem that
is offline).
Finally, \citet{AL18} study a so-called subcommittee
voting, where a final committee is a collection of several subcommittees. Their
model, however, does not take time into account and requires that all
subcommittees are mutually disjoint.

\paragraph{Our Contributions.}
 We present the first work in the multistage model 
 that 
 studies a problem from computational social choice
 and that 
 compares the two cases that we call conservative and revolutionary.
 We prove \cmpv{} and~\rmpv{} to be \NP-complete, 
 even for two agents.
 We present a full parameterized complexity analysis of the two problems 
 (see~\cref{fig:results} for an overview of our results;
 see~\cref{app:fig:results} in the appendix for a more precise overview of the
 results).
 Herein, the central tractability concept is fixed-parameter tractability:
 Given some parameter~$\rho$, a problem is called fixed-parameter tractable (FPT) when it can be solved in $f(\rho) |I|^c$~time,
 where $f$~is a computable function only depending on~$\rho$, $|I|$~is the instance size, and~$c$ is some constant.
 Less positively, we can sometimes only show \XP-membership parameterized by~$\rho$, which means that the problem can be solved
 in polynomial time when $\rho$~is constant (but the degree of the polynomial may depend on~$\rho$).
 Fixed-parameter tractability can often be excluded (under standard parameterized complexity assumptions) showing \W{1}- or \W{2}-hardness
 (using parameterized reductions which are similar to standard polynomial-time many-one reductions).
 \cmpv{} and~\rmpv{} are almost indistinguishable regarding their parameterized complexity,
 but when parameterized by~$\ell$,
 \cmpv{} is \NP-hard and \rmpv{} is contained in~\XP.
 Moreover,
 both problems are contained in~\XP{} and \W{1}-hard regarding the parameter number~$\tau$ of stages;
 Note that for many natural multistage problems (and even temporal graph problems),
 such a classification is unknown---$\tau$
 usually leads to para-\NP-hardness.
 Our results further indicate that efficient data reductions
 (see~\cref{sec:preprocessing} for the definition) 
 in terms of polynomial-size problem kernelizations require a combination with~$\tau$:
 While combining the number of agents with the number of candidates allows for no polynomial-size problem kernel (unless~\NPincoNPslashpoly),
 combining any of the two with~$\tau$ yields kernels of polynomial size.

Many details, marked by~\appsymb, are deferred to the appendix.

\section[Limits of Efficient Computation]{Limits of Efficient Computation}
\label{sec:nphardness}
\appendixsection{sec:nphardness}

To prepare the ground for our algorithmic investigations of the
introduced problems, 
we first settle the computational complexity lower bounds
of~\cmpv{} and~\rmpv{}.
We begin with \NP-hardness for quite restricted cases.

\begin{theorem}[\appref{thm:bothnphard}]
 \label{thm:bothnphard}
 (i) \cmpv{} is \NP-hard even for two agents,
 $\ell=0$,
 $x=1$,
 and~$k=|C|/2$.
 (ii) \rmpv{} is \NP-hard even for two agents, 
 $\ell=2k$, 
 $x=1$, 
 and~$k=|C|/2$.
\end{theorem}

\toappendix
{
\subsection{Proof of~\cref{thm:bothnphard}(i)}
\label{proof:thm:bothnphard}

\noindent
We prove the statements of~\cref{thm:bothnphard} separately,
starting with a specific, computationally intractable case of~\cmpv{},
which follows by a reduction from a \textsc{Vertex Cover} variant.\footnote{We remark
that $x=1$ allows that most of the candidates can have no support.
At least in the conservative setting, however,
there can be situations with small~$\ell$ where this is unavoidable.
Moreover, we can modify the reduction to include at each stage for each candidate
one (additional) supporter and increase~$x$ by~$k$.}

\begin{proposition}%
 \label{prop:cmpvnphard}
 \cmpv{} is \NP-hard even for two agents,
 $\ell=0$,
 $x=1$,
 and~$k=|C|/2$.
\end{proposition}

\noindent
In the following proof,
we give a polynomial-time many-one reduction from a special variant~\prob{Half Vertex Cover} of the \prob{Vertex Cover} problem,
where~$r$ is set to half the number of vertices:
\problemdef{Vertex Cover}
{An undirected graph~$G$ and an integer~$r\in\N$.}
{Is whether there is 
a set~$X\subseteq V(G)$ of size at most~$r$ such that~$G-X$ contains no edge.}
\noindent
It is not difficult to see that \prob{Half Vertex Cover} is \NP-complete:
We can reduce any instance~$(G,r)$ of~\VC{} to \prob{Half Vertex Cover} by adding a clique on~$|V(G)|-2r+2$ vertices to~$G$ if~$r<|V(G)|/2$,
or adding enough isolated vertices to~$G$ until~$r=|V(G)|/2$.

\begin{proof}
 Let~$I=(G=(V,E))$ be an instance of~\prob{Half Vertex Cover}, 
 and let~$E=\{e_1,\ldots,e_m\}$ without loss of generality.
 We construct an instance~$I'=(A,C,u,k,\ell,x)$ of~\cmpv{} in polynomial time as follows.
 
 \noindent\textbf{Construction:}
 Set the set~$C$ of candidates equal to~$V$, 
 and the set~$A$ of agents equal to~$\{a_1,a_2\}$.
 Next,
 construct~$m$ voting profiles as follows.
 For profile~$u_t$, $t\in\set{m}$,
 set~$u_t(a_1)=\{v\}$ and~$u_t(a_2)=\{w\}$, 
 where~$e_t=\{v,w\}$.
 Finally, 
 set~$k=|V|/2$, $\ell=0$, and~$x=1$.
 This finishes the construction.
 
 \noindent\textbf{Correctness:}
 We claim that~$I$ is a \yes-instance if and only if~$I'$ is a \yes-instance.
 
 \RD{}
 Let~$X\subseteq V$ be a vertex cover of size at most~$|V|/2$.
 We claim that the sequence $(C_1,\ldots,C_m)$ with~$C_t=X$ for every~$t\in\set{m}$ is a solution to~$I'$.
 Suppose towards a contradiction that this is not the case.
 Firstly, 
 observe that~$\symdif{C_t}{C_{t+1}}=\emptyset$ for every $t\in\set{m-1}$,
 and that~$|C_t|\leq k$ for every~$t\in\set{m}$.
 Hence, 
 there is a~$t\in\set{m}$ such that~$\score_t(C_t)=0$.
 This means that for edge~$e_t=\{v,w\}$,
 both~$v,w\not\in X$, 
 contradicting the fact that~$X$ is a vertex cover of~$G$.
 Hence,
 $(C_1,\ldots,C_m)$ is a solution to~$I'$.
 
 \LD{}
 Let~$(C_1,\ldots,C_m)$ be a solution to~$I'$.
 Note that since~$\ell=0$,
 $C_i=C_j$ for every~$i,j\in\set{m}$.
 We claim that~$X\ceq C_1$ is a vertex cover of~$G$ (note that~$|X|\leq k$).
 Let~$t\in\set{m}$ be arbitrary but fixed.
 Since~$\score_t(C_t)\geq 1$,
 we know that~$C_t\cap e_t\neq \emptyset$.
 Hence,
 since~$X=C_t$,
 edge~$e_t$ is covered by~$X$.
 Since~$t$ was chosen arbitrarily,
 it follows that~$X$ is a vertex cover of~$G$ of size at most~$k$.
\lqed
\end{proof}
}

\noindent
Herein,
\cref{thm:bothnphard}(ii)
follows from \cref{thm:bothnphard}(i)
due to the following result 
(which we also use later in this section).

\begin{lemma}[\appref{lem:cmpvtormpv}]
 \label{lem:cmpvtormpv}
 There is an algorithm that,
 on every instance $(A,C,U,k,\ell,x)$ with~$\ell=0$ and~$k=|C|/2$ of~\cmpv{},
 computes an equivalent instance~$(A,C',U',k',\ell',x)$
 of~\rmpv{} with~$k'=|C'|/2$, $\ell'=2k'$, and~$|U'|=2|U|+1$ in polynomial time.
\end{lemma}

\appendixproof{lem:cmpvtormpv}
{
\begin{proof}
 Let $I=(A,C,U,k,\ell,x)$ with~$\ell=0$, $k=|C|/2$, and~$\tau$ voting profiles be an instance of~\cmpv{}.
 We construct instance~$I'=(A,C',U',k',\ell',x)$ of~\rmpv{} with~$\ell'=2k'$ in polynomial time as follows.
 
 \struct{Construction:}
 Set the set of candidates~$C'=C\uplus\{z,y\}$,
 where~$z,y$ are new candidates not in~$C$.
 Next, 
 construct~$2\cdot \tau+1$ voting profiles as follows.
 For all~$a\in A$ and all~$t\in\set{\tau}$
 set~$u_{2t-1}'(a)=u_{2t-1}(a)$ and~$u_{2t}'(a)=\{y\}$.
 Moreover,
 set~$u_{2\tau+1}'(a)=\{z\}$ for all~$a\in A$.
 Finally,
 set~$k'=k+1$ and~$\ell'=2k'=|C'|$.
 This finishes the construction.
 
 \struct{Correctness:}
 We claim that $I$ is a \yes-instance if and only if~$I'$ is a \yes-instance.
 
 \RD{}
 Let~$(C_1,\ldots,C_\tau)$ be a solution to~$I$ such that~$|C_1|=k$.
 Since~$\ell=0$, we have that~$C_t=C_{t'}$ for every~$t,t'\in\set{\tau}$.
 Let~$X\ceq C_1\cup\{z\}$ and~$Y\ceq C'\setminus X$.
 Note that~$|X|=|Y|=|C'|/2$.
 We claim that~$(C_1',\dots,C_{2\tau+1}')$ with~$C_{2t-1}'=X$ for every~$t\in\set{\tau+1}$ and $C_{2t}'=Y$ for every~$t\in\set{\tau}$ is a solution to~$I'$.
 Since~$y\in Y$,
 $\score_{2t}(Y)=|A|\geq x$ for every~$t\in\set{\tau}$.
 Since~$u_{2t-1}'(a)=u_{2t-1}(a)$ for every~$a\in A$ and~$t\in\set{\tau}$,
 and~$(C_1,\ldots,C_\tau)$ is a solution to~$I$,
 we have~$\score_{2t-1}(X)\geq x$.
 Moreover,
 since~$z\in X$,
 we have~$\score_{2\tau+1}(X)=|A|\geq x$.
 Lastly,
 as~$\symdif{C_{t}'}{C_{t+1}'}=C'$ for every~$t\in\set{2\tau}$,
 the claim follows.
 
 \LD{}
 Suppose that~$(C_1,\dots,C_{2\tau+1})$  is a solution to~$I'$.
 First observe that,
 due to~$\ell=|C|$, 
 we have that~$\symdif{C_t}{C_{t+1}}=C$ for every~$t\in\set{2\tau}$.
 It follows that~$C_{2t-1}=C_{2t'-1}$ and~$z\in C_{2t-1}$ for every~$t,t'\in\set{\tau+1}$,
 and~$C_{2t}=C_{2t'}$ and~$y\in C_{2t}$ for every~$t,t'\in\set{\tau}$.
 We claim that~$(C_1',C_2',\dots,C_\tau')$ with~$C_i'\ceq C_1\setminus\{z\}$ is a solution to~$I$.
 By construction,
 $\score_i(C_i')\geq x$, 
 and as~$C_i'=C_j'$ and~$|C_i'|\leq k$,
 the claim follows.
\lqed
\end{proof}
}

\toappendix
{
\subsection{Proof of~\cref{thm:bothnphard}(ii)}
\noindent
Combining \cref{lem:cmpvtormpv} and \cref{prop:cmpvnphard},
we get:

\begin{corollary}
 \label{cor:rmpvnphard}
 \rmpv{} is \NP-hard even for two agents, $\ell=2k$, $x=1$, and~$k=|C|/2$.
\end{corollary}

\noindent
\cref{thm:bothnphard} now follows from \cref{prop:cmpvnphard} and \cref{cor:rmpvnphard}.
}

\noindent
We point out that
$\ell=0$ (\cmpv{}) and~$\ell=2k$ (\rmpv{}) are not the only intractable
cases (\appsymb).

\toappendix
{
\subsection{Further Intractable Cases of \cmpv{} and \rmpv{}}
\label{app:ssec:otherell}
We present several intractable cases of \cmpv{} and \rmpv{}, where $0<\ell<2k$.
We also get
the following.

\begin{proposition}%
 \label{prop:nphardnesstwo}
 \cmpv{} with~$\ell=1$
 and
 \rmpv{} with~$\ell=m-1=2k-2$
 are~\NP-hard.
\end{proposition}

\begin{lemma}%
 \label{lem:cmpvnphardellone}
 \cmpv{} is~\NP-hard for~$\ell=1$.
\end{lemma}

{
\begin{proof}
 Let~$I=(A,C,U,k,\ell,x)$ be an instance of~\cmpv{} with~$A=\{a_1,a_2\}$, $x=1$, and~$\ell=0$.
 We construct an instance~$I'=(A'$, $C'$, $U'$, $k'$, $\ell'$, $x')$ of \cmpv{} with~$A'=\{a_1,\dots,a_6\}$,
 $C'=C\cup \{v',v,w\}$,
 $k'=k+2$,
 $\ell'=1$,
 and~$x'=x+4$ in polynomial time as follows.
 
 \struct{Construction:}
 For each~$1\leq t\leq \tau\ceq \tau(U)$,
 let~$a_1,a_2$ approve the same candidates as in~$u_t$.
 If~$t$ is even,
 then let~$a_3,a_4$ approve~$v'$ if~$t$ is divisible by four,
 and approve~$v$ otherwise.
 If~$t$ is odd,
 then let~$a_3,a_4$ approve~$w$.
 Agents~$a_5,a_6$ approve always~$w$.
 This finishes the construction.
 
 \struct{Correctness:}
 We claim that~$I$ is a \yes-instance if and only if~$I'$ is a \yes-instance.
 
 \RD{}
 Let~$(C_1,\dots,C_\tau)$ be a solution to~$C$.
 We claim that the sequence~$(C_1',\dots,C_\tau')$ with $C_t\ceq C_t\cup\{w\}$ if~$t$ is odd,
 $C_t\ceq C_t\cup\{v',w\}$ if~$t$ is even and divisible by four,
 and~$C_t\ceq C_t\cup\{v,w\}$ if~$t$ is even and not divisible by four.
 Observe that~$|C_t|\leq k+2$ and~$\symdif{C_t'}{C_{t+1}'}=\{v'\}$ or $\symdif{C_t'}{C_{t+1}'}=\{v\}$.
 Moreover,
 since in time step~$t$,
 at least one candidate of~$C_t$ is a approved by~$a_1,a_2$,
 by construction we have~$\score_{t}(C_t')\geq 5$.
 This proves the claim.
 
 \LD{}
 Let~$(C_1',\dots,C_\tau')$ be a solution to~$I'$.
 Observe that every solution to~$I'$,
 $w$ must be in every committee,
 $v'$ must be in every committee in an even time step divisible by~four,
 and~$v$ must be in every committee in an even time step not divisible by four.
 Hence, 
 $\symdif{C_t'}{C_{t+1}'}=\{v'\}$ or $\symdif{C_t'}{C_{t+1}'}=\{v\}$ for every~$t\in\set{\tau-1}$.
 We claim that~$(C_1,\dots,C_\tau)$ with~$C_t=C_t'\cap C$ is a solution to~$I$.
 Observe that~$|C_t|\leq k$ and that~$\symdif{C_t}{C_{t+1}}=(\symdif{C_t'}{C_{t+1}'})\cap C=\emptyset$.
 Moreover,
 since for every~$t\in\set{\tau}$ we have that~$\score_t(C_t')\geq 5$,
 at least one of~$a_1,a_2$ must approve a candidate in~$C_t'\cap C=C_t$.
 Thus,
 $\score_t(C_t)\geq 1$,
 and the claim follows.
\lqed
\end{proof}
}

Similarly,
having~$\ell=2k$ is not necessary for~\rmpv{} to be \NP-hard.

\begin{lemma}%
 \label{lem:rmpvnphardell2kminus2}
 \rmpv{} is~\NP-hard for~$\ell=m-1=2k-2$.
\end{lemma}

{
\begin{proof}
 Let~$I=(A,C,U,k,\ell,x)$ be an instance of~\rmpv{} with~$A=\{a_1,a_2\}$, $\ell=2k$, $x=1$, and~$k=|C|/2$.
 We construct an instance~$I'=(A',C',U',k',\ell',x')$ of \cmpv{} with~$A'=\{a_1,\dots,a_4\}$,
 $C'=C\cup \{w\}$,
 $k'=k+1$,
 $\ell'=2k'-2$,
 and~$x'=x+2$ in polynomial time as follows.
 
 \struct{Construction:}
 For each~$1\leq t\leq \tau\ceq \tau(U)$,
 let~$a_1,a_2$ approve the same candidates as in~$u_t$,
 and~$a_3,a_4$ approve~$w$.
 Let~$u_1',\dots,u_\tau'$ denote the obtained voting profiles.
 
 \struct{Correctness:}
 We claim that~$I$ is a \yes-instance if and only if~$I'$ is a \yes-instance.
 
 \RD{} 
 Let~$(C_1,\dots,C_\tau)$ be a solution to instance~$I$.
 We claim that $(C_1',\dots,C_\tau')$ with~$C_t'=C_t\cup\{w\}$ is a solution to~$I'$.
 Note that~$|C_t'|\leq k+1$,
 and that~$|\symdif{C_t'}{C_{t+1}'}|=|\symdif{C_t}{C_{t+1}}|\geq 2k=2k'-2=\ell'$.
 Moreover,
 note that since~$x'=3$,
 for each~$t\in\set{\tau}$ we have that~$a_1,a_2$ approve a candidate from~$C_t'$,
 which is different to~$w$.
 Hence,
 by construction,
 $\score_t(C_t)\geq 1$.
 
 \LD{}
 Let~$(C_1',\dots,C_\tau')$ be a solution to~$I'$.
 Observe that for every solution to~$I'$,
 $w$ must be in every committee.
 Moreover,
 since~$\ell'=|C'|-1=|C|$,
 we have that~$(C_t\cup C_{t+1})\setminus\{w\}=C$.
 We claim that~$(C_1,\dots,C_\tau)$ with~$C_t=C_t'\setminus\{w\}$ is a solution to~$I$.
 Note that~$|C_t|\leq k$,
 $|\symdif{C_t}{C_{t+1}}|=|\symdif{C_t'}{C_{t+1}'}|\geq \ell$,
 and that~$\score_t(C_t)=\score{C_t'}-2\geq x$.
 Thus,
 the claim follows.
\lqed
\end{proof}
}
}

\cref{thm:bothnphard} shows that \cmpv{} and~\rmpv{} remain \nphard{} even for
very specific scenarios.
However, the restrictions
in~\cref{thm:bothnphard} do not deal with
the size~$k$ of committee and the number~$\tau$ of stages.
This gives hope that instances in which these numbers are small
could be
solved more effectively.
However, as we show in the remainder of this section, regarding parameters~$k$
and~$\tau$ (and their combination) for~\cmpv{} and parameter~$\tau$ for \rmpv{}
we (presumably) cannot obtain running times for which the exponential blow-up is
only depending on values of the parameters.

\begin{theorem}[\appsymb]
 \label{thm:cmpvwhardktau}
 \cmpv{} is 
 \begin{inparaenum}[(i)]
    \item \wonehard{} when parameterized by~$k+\tau$, 
 even if~$\ell=0$;
 \item \W{2}-hard when parameterized by~$k$, even if~$x=1$ and~$\ell=0$.
 \end{inparaenum}

\end{theorem}

\noindent
For~\cref{thm:cmpvwhardktau}(ii),
we reduce from the~\W{2}-complete \prob{Dominating Set} problem (\apprefX{app:cmpvhardjtau}).%
\toappendix
{
\subsection{Proof of \cref{thm:cmpvwhardktau}(ii)}
\label{app:cmpvhardjtau}
\begin{proof}
  Let~$I=(G,k)$ with~$G=(V,E)$ be an instance of~\prob{Dominating Set},
  where~$V=\{v_1,\dots,v_n\}$ without loss of generality.
  We construct an instance~$I'\ceq (A,C,U,k',\ell,x)$ with~$k'\ceq k$,
  $\ell\ceq 0$, 
  $x\ceq 1$, and~$\tau=n$ as follows.
  Let~$A\ceq \{a_1,\dots,a_{\Delta+1}\}$, 
  where~$\Delta$ denotes the maximum degree of~$G$.
  Let~$C\ceq V$.
  Now, 
  for every~$t\in\set{\tau}$,
  set~$u_t$ such that~$u_t(A)=N_G[v_t]$,
  that is,
  the closed neighborhood of~$v_t$ in~$G$.
  This finishes the construction.
  We claim that~$I$ is \yes-instance if and only if~$I'$ is a \yes-instance.
  
  \RD{}
  Let~$X\subseteq V$ be a dominating set of~$G$.
  We claim that~$(S_1,\dots,S_\tau)$ with~$S_i=X$ for all~$i\in\set{\tau}$ is solution to~$I'$.
  Observe that~$|N_G[v_t]\cap X|\geq 1$,
  and hence~$|u_t(A)\cap S_t|\geq 1$.
  
  \LD{}
  Let~$(S_1,\dots,S_\tau)$ be a solution to~$I'$.
  Note that since~$\ell=0$,
  $S_i=S_j$ for all~$i,j\in\set{\tau}$.
  Let~$X\ceq S_1$.
  Note that~$|X|\leq k$.
  We claim that~$X$ is a dominating set.
  Suppose not,
  then there is a~$t\in\set{n}$ such that~$X\cap N_G[v_t]=\emptyset$.
  Then, $|S_t\cap N_G[v_t]|= |S_t\cap u_t(A)| = 0 < x$,
  a contradiction to the fact that $(S_1,\dots,S_\tau)$ is a solution.
\end{proof}
}%
In the reduction behind the proof of~\cref{thm:cmpvwhardktau}(i),
we employ \emph{Sidon sets} defined subsequently.
 A Sidon set is a set~$S=\{s_1, s_2, \ldots, s_b\}$ of $b$~natural numbers such
 that every pairwise sum of the elements in~$S$ is different.%
\newcommand{\sidonSet}{\ensuremath{S}}
\newcommand{\sidonElem}{\ensuremath{s}}
Sidon sets can be computed efficiently.

\begin{lemma} \label{lem:sidon-comp}
 A Sidon set of size~$b$ can be computed in $\O(b)$~time if~$b$ is encoded in
 unary.
\end{lemma}
\begin{proof}
 Suppose we aim at obtaining a Sidon set~$\sidonSet{}=\{\sidonElem_1, \ldots,
 \sidonElem_b\}$. For every~$i \in \set{b}$, we
 compute~$\sidonElem{}_i\ceq 2\hat{b}i+ (i^2\bmod \hat{b})$, where $\hat{b}$~is the
 smallest prime number greater than~$b$~\citep{ET41}. Thus, given~$\hat{b}$, one
 can compute~\sidonSet{} in linear time.

 It remains to show how to find~$\hat{b}$ in linear time. Due to the
 Bertrand-Chebyshev~\citep{Chebyshev1852} theorem, we have that $\hat{b} < 2b$.
 Searching all prime numbers smaller than~$2b$ is doable in $\O(b)$~time (see,
 for example, an intuitive algorithm by \citet{GM78}).
\lqed
\end{proof}

\newcommand{\cGraph}{\ensuremath{G}}
\newcommand{\cVerticesOf}[1]{\ensuremath{V_{#1}}}
\newcommand{\cEdge}{\ensuremath{e}}
\newcommand{\cEdges}{\ensuremath{E}}
\newcommand{\cEdgesOf}[2]{\ensuremath{E_{#2}^{#1}}}
\newcommand{\cCliqueSize}{\ensuremath{q}}
\newcommand{\cClique}{\ensuremath{K}}
\newcommand{\cVertex}{\ensuremath{v}}
\newcommand{\cVvertex}{\ensuremath{u}}
\newcommand{\instanceFrom}{\ensuremath{\hat{I}}}
\newcommand{\instanceTo}{\ensuremath{I}}
\newcommand{\requiredApprovals}{\ensuremath{x}}
\newcommand{\prefProfileOf}[1]{\ensuremath{p({#1})}}
\newcommand{\pprefProfileOf}[1]{\ensuremath{p'({#1})}}
\newcommand{\idOp}{\ensuremath{\mathop{\mathrm{id}}}}
\newcommand{\idOf}[1]{\ensuremath{\idOp({#1})}}
\newcommand{\changeBound}{\ensuremath{\ell}}
\newcommand{\committeeSize}{\ensuremath{k}}
\newcommand{\committee}{\ensuremath{C}}
\newcommand{\gadget}{\ensuremath{\mathcal{G}}}
\newcommand{\vertexTotCount}{\ensuremath{h}}
\begin{proof}[Proof of~\cref{thm:cmpvwhardktau}(i)]
 We reduce from~\multicoloredClique{} that is \wone-complete when parameterized
 by the solution size~\citep{P03,FHRV09}. An instance~\instanceFrom{} of~\multicoloredClique{}
 consists of a \cCliqueSize{}-partite graph~$\cGraph=(\cVerticesOf{1} \uplus 
 \cVerticesOf{2} \uplus\dots\uplus \cVerticesOf{\cCliqueSize}, \cEdges)$ and the task
 is to decide whether there is a set~\cClique{} of~\cCliqueSize{} pairwise
 connected vertices, each from a distinct part.
 For brevity, for some~$i, j \in \set{\cCliqueSize}$, $i < j$,
 let~\cEdgesOf{j}{i} be a set of edges connecting vertices from parts~$\cVerticesOf{i}$ and~$\cVerticesOf{j}$; thus, $\cEdges{}=\bigcup_{i,j \in \set{\cCliqueSize}, i < j}
  \cEdgesOf{j}{i}$.

 \smallskip\noindent\textbf{Construction.} In the corresponding instance~\instanceTo{}
 of~\cmpv{}, we let all vertices and edges in~\cGraph{} be candidates. Then, we
 define three gadgets (see~\cref{fig:wonehardness} for an illustration):
 \begin{figure*}[t]
  \centering
    \begin{tikzpicture}

      \usetikzlibrary{calc,positioning}
      \usetikzlibrary{decorations.pathreplacing}

      \def\xr{0.8}
      \def\yr{0.7}

      \def\bxw{6}
      \def\bxxw{0.6}
      \def\bxh{4.5}
      \def\bxxh{0.925*\bxh}
      \def\vdsc{0.8}

      \newcommand{\decbox}[4]{
        \node[above =of #1,yshift=-\yr*1.35cm] {#2};
        \node[above =of #1.south,yshift=-\yr*1.45cm,scale=\vdsc] {$\vdots$};
        \node[below =of #1.north,yshift=\yr*1.6cm,scale=\vdsc] {$\vdots$};
        \node[above =of #1.center,yshift=-\yr*0.7cm,scale=\vdsc] {$\vdots$};
        \node[below =of #1.center,yshift=\yr*0.75cm,scale=\vdsc] {$\vdots$};
        \node[above =of #1.center,yshift=-\yr*1.2cm,font=\small] (a1) {#3};
        \node at (#1.center)[scale=\vdsc]{$\vdots$};
        \node[below =of #1.center,yshift=\yr*1.0cm,font=\small] (a2) {#3};
        \draw [decorate,decoration={brace,amplitude=4pt},xshift=-2cm,yshift=10pt] (a2.west) -- (a1.west) node [black,midway,xshift=-0.2cm,anchor=east,font=\small]{#4};	
      }

      \begin{scope}
        \node (bx) at (0,0)[minimum width=\xr*\bxw*0.8 cm, minimum height=\yr*\bxh cm, dashed, rounded corners, draw]{};
          \node[below =of bx,yshift=\yr*1.45cm,font=\small] {Vertex selection gadget};
        \node[above =of bx,yshift=-\yr*1.35cm] {$\cdots$};
        \node[above =of bx,yshift=-\yr*1.35cm,xshift=-\xr*1.75cm] {$\cdots$};
        \node[above =of bx,yshift=-\yr*1.35cm,xshift=\xr*1.75cm] {$\cdots$};

        \node (bxxi) at (0-1*\xr,0)[minimum width=\xr*\bxxw cm, minimum height=\yr*\bxxh cm, densely dotted, rounded corners, draw]{};
          \decbox{bxxi}{$V_i$}{$v$}{$x$~$\times$}

        \node (bxxj) at (0+1*\xr,0)[minimum width=\xr*\bxxw cm, minimum height=\yr*\bxxh cm, densely dotted, rounded corners, draw]{};
        \decbox{bxxj}{$V_j$}{$w$}{$x$~$\times$}
      \end{scope}

      \begin{scope}[xshift=0.725*\xr*\bxw cm]
        \node (bx) at (0,0)[minimum width=\xr*\bxw*0.55 cm, minimum height=\yr*\bxh cm, dashed, rounded corners, draw]{};
        \node[below =of bx,yshift=\yr*1.45cm,font=\small] {Edge selection gadget};
        \node[above =of bx,yshift=-\yr*1.35cm,xshift=-\xr*1.25cm] {$\cdots$};
        \node[above =of bx,yshift=-\yr*1.35cm,xshift=\xr*1.25cm] {$\cdots$};

        \node (bxe) at (0,0)[minimum width=\xr*\bxxw cm, minimum height=\yr*\bxxh cm, densely dotted, rounded corners, draw]{};
          \decbox{bxe}{$E_i^j$}{$e$}{$x$~$\times$};
      \end{scope}

      \newcommand{\decboxc}[5]{
        \node[above =of #1,yshift=-\yr*1.35cm,font=\footnotesize] {#2};
        \%
        \node[below =of #1.north,yshift=-\yr*1.7cm,font=\small] (a1) {#3};
        \node[below =of #1.north,yshift=-\yr*1.0cm,scale=\vdsc] {$\vdots$};
        \node[below =of #1.north,,yshift=-\yr*0.9cm,font=\small] (a2) {#3};
        \node[below =of #1.north,yshift=\yr*0.3cm,font=\small] (b1) {#5};
        \node[below =of #1.north,yshift=\yr*0.2cm,scale=\vdsc]{$\vdots$};
        \node[below =of #1.north,yshift=-\yr*0.5cm,font=\small] (b2) {#5};
        \node[below =of #1.north,yshift=\yr*1.5cm,font=\small] (c1) {#4};
        \node[below =of #1.north,yshift=\yr*1.4cm,scale=\vdsc] {$\vdots$};
        \node[below =of #1.north,yshift=\yr*0.7cm,font=\small] (c2) {#4};
        \node[below =of #1.north,yshift=-\yr*1.9cm,scale=\vdsc] {$\vdots$};
      }

      \newcommand{\decboxcA}[3]{    
        \draw [decorate,decoration={brace,amplitude=4pt},xshift=-2cm,yshift=10pt] (a1.west) -- (a2.west) node [black,midway,xshift=-0.2cm,anchor=east,font=\small]{#1};	
        \draw [decorate,decoration={brace,amplitude=4pt},xshift=-2cm,yshift=10pt] (b2.west) -- (b1.west) node [black,midway,xshift=-0.2cm,anchor=east,font=\small]{#2};	
        \draw [decorate,decoration={brace,amplitude=4pt},xshift=-2cm,yshift=10pt] (c2.west) -- (c1.west) node [black,midway,xshift=-0.2cm,anchor=east,font=\small]{#3};	
      }

      \newcommand{\decboxcB}[3]{    
        \draw [decorate,decoration={brace,amplitude=4pt},xshift=-2cm,yshift=10pt] (a2.east) -- (a1.east) node [black,midway,xshift=0.2cm,anchor=west,font=\small]{#1};	
        \draw [decorate,decoration={brace,amplitude=4pt},xshift=-2cm,yshift=10pt] (b1.east) -- (b2.east) node [black,midway,xshift=0.2cm,anchor=west,font=\small]{#2};	
        \draw [decorate,decoration={brace,amplitude=4pt},xshift=-2cm,yshift=10pt] (c1.east) -- (c2.east) node [black,midway,xshift=0.2cm,anchor=west,font=\small]{#3};	
      }

      \begin{scope}[xshift=1.725*\xr*\bxw cm]
        \node (bx) at (0,0)[minimum width=\xr*\bxw*1.35 cm, minimum height=\yr*\bxh cm, dashed, rounded corners, draw]{};
        \node[below =of bx,yshift=\yr*1.45cm,font=\small] {Coherence gadget};
        \node[above =of bx,yshift=-\yr*1.35cm,xshift=-\xr*2.5cm] {$\cdots$};
        \node[above =of bx,yshift=-\yr*1.35cm,xshift=\xr*2.5cm] {$\cdots$};

        \node (bxp) at (0-0.6*\xr,0)[minimum width=\xr*\bxxw cm, minimum height=\yr*\bxxh cm, densely dotted, rounded corners, draw]{};
        \decboxc{bxp}{$p((i,j))$}{$e$}{$v$}{$w$}
	\decboxcA{$x-\idOf{v,w}$~$\times$}{$\idOf{w}$~$\times$}{$\idOf{v}$~$\times$};

        \node (bxpp) at (0+1*\xr,0)[minimum width=\xr*\bxxw cm, minimum height=\yr*\bxxh cm, densely dotted, rounded corners, draw]{};
        \decboxc{bxpp}{$p'((i,j))$}{$e$}{$v$}{$w$}
	\decboxcB{$\times$~$\idOf{v,w}$}{$\times$~$\frac{x}{2}-\idOf{w}$}{$\times$~$\frac{x}{2}-\idOf{v}$};
      \end{scope}

      \draw[very thick,<->,>=latex] (-0.5*\bxw*\xr,-\yr*\bxh/2) to node[midway,sloped,yshift=0.4*\xr cm]{agents / voting profiles}(-0.5*\bxw*\xr,\yr*\bxh/2);

      \draw[very thick,->,>=latex] (-0.4*\bxw*\xr,1.35*\yr*\bxh/2) to
      node[midway,above,sloped,yshift=-3pt]{stages}(2.4*\bxw*\xr,1.35*\yr*\bxh/2);

      \end{tikzpicture}
  \caption{Illustration of the construction in the proof of~\cref{thm:cmpvwhardktau},
  exemplified with edge~$e=\{v,w\}\in E_i^j$ with~$v\in V_i$ and~$w\in V_j$.
  A column represents a stage (which in turn represents an vertex or edge selection gadget for some vertex or edge set, respectively, or a coherence gadget of a pair of colors)
  and a row represents an agent (approving either a vertex or an edge). For
  brevity, we use~$\idOf{v,w}$ to denote~$\idOf{v} + \idOf{w}$.}
  \label{fig:wonehardness}
 \end{figure*}
 the \emph{vertex selection} gadget, the \emph{edge
 selection} gadget, and the \emph{coherence} gadget. Further, we show how to use
 the gadgets to construct~\instanceTo{}. Instance~\instanceTo{} will be
 constructed in a way that its solution is a single committee of size
 exactly~$\cCliqueSize + {\cCliqueSize \choose 2}$ corresponding to vertices and
 edges of a clique witnessing a \yes-instance of~\instanceFrom{} (if one exists).
 To define the gadgets, 
 we use a value~\requiredApprovals{} that we explicitly
 define at the end of the construction.

 \paragraph{Vertex selection gadget.}
 Fix some part~$\cVerticesOf{i}$, $i \in \set{\cCliqueSize}$. The vertex selection gadget
 for~$i$ ensures that exactly one vertex from~$\cVerticesOf{i}$ is selected.
 We construct the gadget by forming a preference
 profile~\prefProfileOf{\cVerticesOf{i}}
 consisting of~$\requiredApprovals \cdot
 |\cVerticesOf{i}|$ agents such that each vertex~$\cVertex \in \cVerticesOf{i}$
 is approved by exactly \requiredApprovals{}~agents.
 
 \paragraph{Edge selection gadget.}
 For each two parts~$\cVerticesOf{i}$ and~$\cVerticesOf{j}$ such that~$i<j$, we construct the
 edge selection gadget that allows to select exactly one edge
 from~\cEdgesOf{j}{i}. Accordingly, we build a preference
 profile~\prefProfileOf{\cEdgesOf{j}{i}} consisting of~$\requiredApprovals \cdot
 |\cEdgesOf{j}{i}|$ agents. Again, each edge in~$\cEdgesOf{j}{i}$ is approved by
 exactly \requiredApprovals{}~agents.

 \paragraph{Coherence gadget.}
 For the construction of the coherence gadget, let~$\vertexTotCount \ceq 
 |\bigcup_{i \in \set{\cCliqueSize}} \cVerticesOf{i}|$ and
 let~$\sidonSet{}=\{\sidonElem_1, \ldots, \sidonElem_\vertexTotCount\}$ be a
 Sidon set computed according to~\cref{lem:sidon-comp}. We define a bijection~$\idOp \colon \bigcup_{i \in \set{\cCliqueSize}} \cVerticesOf{i}
 \rightarrow S$ associating each vertex of~\cGraph{} with its (unique)
 \emph{id}. Now, the construction of the coherence gadget for some
 pair~$\{\cVerticesOf{i},\cVerticesOf{j}\}$ of parts such that~$i < j$ goes as
 follows. We introduce two preference profiles~\prefProfileOf{(i,j)}
 and~\pprefProfileOf{(i,j)}. In preference profile~\prefProfileOf{(i,j)}, (i)
 each candidate~$\cVertex\in \cVerticesOf{i} \cup \cVerticesOf{j}$ is approved
 by exactly \idOf{\cVertex}~agents and (ii) each edge~$\cEdge=\{\cVertex,
 \cVertex'\} \in \cEdgesOf{j}{i}$ is approved by exactly~$(\requiredApprovals -
 \idOf{\cVertex} - \idOf{\cVertex'})$~agents. In preference
 profile~\pprefProfileOf{(i,j)}, (i) each candidate~$\cVertex\in \cVerticesOf{i}
 \cup \cVerticesOf{j}$ is approved by exactly $\frac{\requiredApprovals}{2} -
 \idOf{\cVertex}$~agents and (ii) each edge~$\cEdge=\{\cVertex, \cVertex'\}\in
 \cEdgesOf{j}{i}$ is approved by exactly $(\idOf{\cVertex} +
 \idOf{\cVertex'})$~agents.

 Having all the gadgets defined it remains to use them to form the agents and
 the preference profiles of instance~\instanceTo{}; and to
 define~\requiredApprovals{}, \changeBound, and~\committeeSize{}. Since we want
 to have a committee consisting of $\cCliqueSize$~vertices  and~${\cCliqueSize \choose 2}$~edges, we
 let~$\committeeSize{} \ceq  \cCliqueSize + {\cCliqueSize \choose 2}$. We aim at a single committee, thus
 we set~$\changeBound{} = 0$, which enforces that the committee must stay the
 same over time. Further, we set~$\requiredApprovals=2\sidonElem_\vertexTotCount$.
 Finally, to form the preference profiles of~\instanceTo{} we put together, in
 any order, vertex selection gadgets for every part~$\cVerticesOf{i}$, $i \in \set{\cCliqueSize}$
 as well as edge selection gadgets and coherence gadgets for every pair~$\{\cVerticesOf{i},\cVerticesOf{j}\}$
 of parts such that~$i<j$. As for the agents of~\instanceTo{}, with each
 gadget~$\gadget$ we add a separate set of agents needed to implement~$\gadget$
 making sure that all other agents introduced by all other gadgets are approving
 no candidate in their voting profiles occurring in~$\gadget$.

 The running time analysis
 and correctness proof can be found in~\cref{app:cmpvwhardktau} (\appsymb).
 \appendixproof{thm:cmpvwhardktau}{
 \noindent\textbf{Running Time.}\label{app:cmpvwhardktau}
 The reduction builds a polynomial number (with
 respect to the input size) of copies of gadgets. However, the construction time
 of the coherence gadget heavily depends on the computation of Sidon
 set~\sidonSet{} (of size~\vertexTotCount) and the value of its largest
 element~$\sidonElem_\vertexTotCount$. Since~\vertexTotCount{} is linearly
 bounded in the size of the input, due to~\cref{lem:sidon-comp}, we get a
 polynomial running time (with respect to the input size) of
 computing~\sidonSet{}. As a by-product we get that the largest element
 of~\sidonSet{} is also polynomially upper-bounded (with respect to the input
 size).

 \noindent\textbf{Correctness.}
 Naturally, if instance~\instanceTo{} of~\cmpv{} is a \yes-instance, then there
 is a committee that is witnessing this fact; otherwise, such a committee does
 not exist. Since the candidates are all vertices and edges of graph~\cGraph{},
 we refer to edges and vertices being part of some committee as, respectively,
 \emph{selected} vertices and edges.

 On the way to prove the correctness of the reduction, we state the following
 lemma about the coherence gadget.
 \begin{lemma}%
  \label{lem:coherence_gadget}
  For a pair~$\{\cVerticesOf{i},\cVerticesOf{j}\}$ of parts such that~$i<j$, let~\committee{} be a committee
  selecting exactly one vertex from either of them, $\cVertex$, $\cVertex'$
  respectively, and exactly one edge~$\cEdge{} \in \cEdgesOf{j}{i}$
  connecting some vertices of parts~$\cVerticesOf{i}$ and~$\cVerticesOf{j}$. Then, the scores
  of~\committee{} for the profiles of the coherence gadget for parts~$\cVerticesOf{i}$ and~$\cVerticesOf{j}$
  are at least~\requiredApprovals{} if and only if~$\cEdge=\{\cVertex,
  \cVertex'\}$.
 \end{lemma}
 
 \begin{proof}
  Assume that a selected edge connects vertices~$\cVvertex$ and~$\cVvertex'$
  from, respectively, part~$\cVerticesOf{i}$ and part~$\cVerticesOf{j}$. Let us compute the scores
  of~\committee{} for profiles~\prefProfileOf{(i,j)}
  and~\pprefProfileOf{(i,j)} of the coherence gadget for~$i$ and~$j$ (note that
  due to assumptions on~\committee{} only candidates~\cVertex{},
  $\cVertex'$, and~\cEdge{} contribute to the scores):
  \begin{align}
   &\score_{\prefProfileOf{(i,j)}}(\committee) = \idOf{\cVertex} +
   \idOf{\cVertex'} + \requiredApprovals - \idOf{\cVvertex} - \idOf{\cVvertex'},
   \label{eq:edge_selection_one}
   \\
   &\score_{\pprefProfileOf{(i,j)}}(\committee) =
   \frac{\requiredApprovals}{2}-\idOf{\cVertex}
   +\frac{\requiredApprovals}{2}-\idOf{\cVertex'} + \idOf{\cVvertex} +
   \idOf{\cVvertex'}.
   \label{eq:edge_selection_two}
  \end{align}

  If both of the scores are at least~\requiredApprovals{}, after
  simplifying the equations we arrive at:
  \begin{align*}
   &\idOf{\cVertex} + \idOf{\cVertex'} \geq \idOf{\cVvertex} + \idOf{\cVvertex'}
   \\ \land &
   \idOf{\cVvertex} + \idOf{\cVvertex'} \geq \idOf{\cVertex}
   + \idOf{\cVertex'}.
  \end{align*}
  This, in turn, simply means that:
  \begin{equation}\label{eq:vertices_equality}
   \idOf{\cVvertex} + \idOf{\cVvertex'} = \idOf{\cVertex} + \idOf{\cVertex'}.
  \end{equation}
  Recall that the image of bijective function~$\idOf{\cdot}$ is a Sidon set.
  Thus, by the definition of the Sidon set, \cref{eq:vertices_equality} is true
  if and only if~$\{\cVertex, \cVertex'\} = \{\cVvertex, \cVvertex'\} = \cEdge$.

  The opposite direction follows directly from the definition of Sidon sets
  because all~$\idOf{\cdot}$-terms in~\cref{eq:edge_selection_one}
  and~\cref{eq:edge_selection_two} cancel out resulting in both scores
  being~$\requiredApprovals$.
 \lqed
\end{proof}

 Analogously, we provide the following lemma that describes the role of the
 vertex and edge selection gadgets.
 
 \begin{lemma}%
  \label{lem:selection_gadgets}
  Every committee witnessing a \yes-instance selects exactly one vertex from
  each part of~\cGraph{} and exactly~$\cCliqueSize \choose 2$ edges, one for
  each distinct pair~$\{\cVerticesOf{i},\cVerticesOf{j}\}$ of parts.
  Additionally, the scores of such a committee~\committee{} in all vertex and
  edge selection gadgets are exactly~\requiredApprovals{}.
 \end{lemma}
 
 \begin{proof}
  Let us fix a committee~$\committee$ witnessing a \yes-instance. Towards a
  contradiction, assume that there exists a part~\cVerticesOf{i} a vertex of which is not
  selected. Then, since in profile~\prefProfileOf{V_i} only vertices from
  part~\cVerticesOf{i} are approved, the score of~\committee{} in~\prefProfileOf{V_i} is
  zero; the contradiction. Similarly, assume that there is a pair~$\{\cVerticesOf{i},\cVerticesOf{j}\}$ of
  parts such that no edge in~$\cEdgesOf{j}{i}$ is selected by~\committee{}. Again,
  by an analogous argument, if this is the case, then the scores of~\committee{}
  in profiles~\prefProfileOf{\cEdgesOf{j}{i}} and~\pprefProfileOf{\cEdgesOf{i}{j}}
  are zero which gives a contradiction. From the fact that there are exactly
  \cCliqueSize{}~vertex selection gadgets and exactly $\cCliqueSize \choose
  2$~edge selection gadgets, it follows that~\committee{} selects exactly one
  vertex from each part of~\cGraph{} and exactly one edge for a pair of parts.
  By the construction of the vertex and edge selection gadgets, indeed each of
  them gives score exactly~\requiredApprovals{}.
 \lqed
\end{proof}

 Having the above lemmas, we finally show the correctness of our reduction. In
 one direction, suppose \cClique{}~is a clique of size~\cCliqueSize{} in
 graph~\cGraph{} of an instance~\instanceFrom{} of~\multicoloredClique{}. We
 construct a committee~\committee{} selecting all vertices of~\cClique{} and
 all edges connecting the vertices in~\cClique{}. Indeed, the necessary
 condition from~\cref{lem:selection_gadgets} is met because of the definition of
 a clique and the fact that~\cClique{} is a clique of size~\cCliqueSize{}. Again
 by the definition of a clique, for every two vertices from distinct parts,
 committee~\committee{} selects the edge that connects them; thus,
 by~\cref{lem:coherence_gadget}, committee~\committee{} is witnessing the
 positive answer.
 
 For the opposite direction, suppose there is no clique in~\cGraph{}
 of~\instanceFrom{}. We show that there is no committee~\committee{} witnessing
 a \yes-instance. Due to~\cref{lem:selection_gadgets} it follows
 that~\committee{} selects a vertex for each part of~\cGraph{} and an edge for
 each pair of parts.
 However, since there is no clique in~\cGraph{} every
 committee~\committee{} there exists
 at least one pair~$\{\cVertex, \cVertex'\}$ of vertices from distinct parts
 that are not connected with an edge. Thus, due to~\cref{lem:coherence_gadget}
 at least in one coherence gadget there is at least one profile for which the
 score of~\committee{} is below~\requiredApprovals{}; which finishes the
 argument.
\lqed
}
\end{proof}

\noindent
As for \rmpv{} and the parameters~$k$ and~$\tau$, the situation
differs from that for~\cmpv{}. Namely, for the former we can only show the
(parameterized) intractability with respect to~$\tau$.

\begin{theorem}[\appref{thm:rmpvwhardtau}]
 \label{thm:rmpvwhardtau}
 \rmpv{} parameterized by~$\tau$ is \W{1}-hard.
\end{theorem}

\appendixproof{thm:rmpvwhardtau}
{
To prove~\cref{thm:rmpvwhardtau},
we aim for employing~\cref{lem:cmpvtormpv},
which forms a parameterized reduction regarding~$\tau$.
To this end,
we need firstly to reduce any instance with~$\ell=0$ of~\cmpv{} to an equivalent instance of~\cmpv{} with~$k=|C|/2$ and~$\ell=0$ such that the number of resulting voting profiles only depends on~$\tau$.

\begin{lemma}
 \label{lem:rmpvtormpv}
 There is an algorithm that,
 on every instance~$(A,C,\allowbreak U,k,\ell,x)$ with~$\ell=0$ of~\cmpv{},
 computes in polynomial time an equivalent instance~$(A',C',U',k',\ell',x')$ of~\cmpv{} with~$k'=|C'|/2$, $\ell'=0$, and~$|U'|=|U|$.
\end{lemma}

\begin{proof}
 Let~$I=(A,C,U,k,\ell,x)$ with~$\ell=0$ be an instance of~\cmpv{}.
 We construct an instance $I'=(A',C',U',k',\ell',x')$ of~\cmpv{} with~$k'=|C'|/2$, $\ell'=0$, and~$|U'|=|U|$, as follows.
 
 \struct{Construction:}
 Initially,
 we set~$I'$ to~$I$.
 If~$k=|C|/2$, we are done.
 If~$k>|C|/2$,
 then we add~$2k-|C|$ candidates to~$C$, forming~$C'$.
 If~$k<|C|/2$,
 then we add~$|C|-2k$ candidates to~$C$, say set~$C^*$,
 forming~$C'=C\cup C^*$.
 Moreover,
 we add~$|A|\cdot (|C|-2k)$ agents to~$A$, say~$A^*$,
 forming~$A'=A\cup A^*$,
 where for each~$c\in C^*$,
 in each stage, 
 exactly~$|A|$ of the new agents approve~$c$ (this forms~$U'$).
 Note that~$|U'|=|U|$.
 Set~$k'=|C'|/2$, $\ell'=0$,
 and~$x'=x+|A|(|C|-2k)$.
 Since correctness for the cases~$k=|C|/2$ and~$k>|C|/2$ is immediate,
 we prove in the following correctness for the case of~$k<|C|/2$.
 
 \struct{Correctness:}
 In the case of~$k<|C|/2$,
 we claim that~$I$ is a \yes-instance if and only if~$I'$ is a \yes-instance.
\yes-instance 
 \RD{}
 Let~$\calC=(C_1,\ldots,C_\tau)$ be a solution to~$I$.
 We claim that~$(C_1',\dots,C_\tau')$ with~$C_i'=C_i\cup C^*$ is a solution to~$I'$.
 First note that~$|C_t'|\leq |C_t|+|C|-2k \leq |C|-k=|C'|/2$.
 Moreover, 
 we have that~$|\symdif{C_t'}{C_{t+1}'}|=|\symdif{C_t}{C_{t+1}}|=0$.
 Finally,
 we have that~$\score_t(C_t')=\score_t(C_t)+|A|(|C|-2k)\geq x+|A|(|C|-2k)=x'$.
 
 \LD{}
 Let~$\calC'=(C_1',\dots,C_\tau')$ be a solution to~$I'$.
 First observe that~$C^*\subseteq C_t'$ for all~$i\in\set{\tau}$.
 Suppose not,
 that is,
 there is a~$t\in\set{\tau}$ and~$c\in C^*$ such that~$c\not\in C_t'$.
 Then the score of~$C_t'$ is at most~$|A|+|A|(|C|-2k-1)=|A|(|C|-2k)<x'$,
 a contradiction to the fact that~$\calC'$ is a solution.
 We claim that~$(C_1,\ldots,C_\tau)$ with~$C_t=C_t^*\setminus C^*$ a solution to~$I$.
 Note that~$|C_t|=|C_t'|-(|C|-2k)\leq (2|C|-2k)/2 -(|C|-2k)=k$,
 and that~$|\symdif{C_t}{C_{t+1}}|=|\symdif{C_t'}{C_{t+1}'}|=0$.
 Finally,
 note that~$\score_t(C_t) = \score_t(C_t')-|A|(|C|-2k)\geq x+|A|(|C|-2k)-|A|(|C|-2k)=x$.
\lqed
\end{proof}

\begin{proof}[Proof of~\cref{thm:rmpvwhardtau}]
 \cref{lem:rmpvtormpv} followed by~\cref{lem:cmpvtormpv} gives a parameterized reduction from~\cmpv{} to \rmpv{} regarding the parameter~$\tau$.
 \cref{thm:cmpvwhardktau} then finishes the proof.
\lqed
\end{proof}
}
\noindent
Altogether, \cref{thm:rmpvwhardtau,thm:bothnphard,thm:cmpvwhardktau} mark
clear borders of computational tractability, 
allowing us to refine our search for
efficient (parameterized) algorithms for the problems we introduced.

\section[Polynomial for Constant Parameter Values]{Polynomial for Constant
Parameter Values}%
\label{sec:xp-and-fpt}
\appendixsection{sec:xp-and-fpt}

We start a series of algorithms in this section, with the one offering a
polynomial-time running time in the case of a small target committee size~$k$.
The proof of \cref{thm:xpregardingk} is based on computing in \XP-running time
an auxiliary directed graph in which we then check for the existence of an
$s$-$t$ path witnessing a \yes-instance.

\appendixproof{thm:xpregardingk}%
{%
\begin{proof}%
 Let $I=(A,C,U,k,\ell,x)$ be an instance of~\cmpv{} (or \rmpv{}) with~$n$ agents, 
 $m$ candidates,
 and~$\tau$ voting profiles.
 We compute the following directed graph~$D=(V,A)$ with vertex set~$V=\{s,t\}\uplus V^1\uplus\dots\uplus V^\tau$,
 where~$V^t=\{v^t_X\mid X\subseteq C, |X|\leq k, \score_t(X)\geq x\}$,
 and arc set~$A$ composed of the sets~$\{(s,v)\mid v\in V^1\}$, 
 $\{(v,t)\mid v\in V^\tau\}$, 
 and for each~$t\in\set{\tau-1}$ the sets
 $\{(v_X^t,v_Y^{t+1})\mid |\symdif{X}{Y}|\leq \ell\}$ in the conservative variant
 (and $\{(v_X^t,v_Y^{t+1})\mid |\symdif{X}{Y}|\geq \ell\}$ in the revolutionary variant).
 Note that there are at most~$\tau\cdot m^k+2$ vertices and at most~$(\tau-1)\cdot (m^k)^2+2m^k$ arcs.
 Hence,
 $D$ can be constructed in~$\O(\tau\cdot k^2 m^{2k+1}n)$ time.
 We claim that~$I$ is a \yes-instance if and only if~$D$ admits an~$s$-$t$ path.
 Note that we can check for an~$s$-$t$ path in~$D$ in time linear in the size of~$D$.
 
 \RD{}
 Let~$(C_1,\dots,C_\tau)$ be a solution to~$I$.
 We claim that~$P=(s,v_{C_1}^1,\dots,v_{C_\tau}^\tau,t)$ is an $s$-$t$ path in~$D$.
 Clearly,
 the arcs~$(s,v_{C_1}^1)$ and~$(v_{C_\tau}^\tau,t)$ are in~$A$.
 For each~$t\in\set{\tau-1}$,
 the arc~$(v_{C_t}^t,v_{C_{t+1}}^{t+1})$ exists since~$|\symdif{C_t}{C_{t+1}}|\leq \ell$ ($\geq \ell$ in the revolutionary case).
 Hence,
 $P$ is an~$s$-$t$ path in~$D$.

 \LD{}
 Let~$P=(s,v_{C_1}^1,\dots,v_{C_\tau}^\tau,t)$ be an $s$-$t$ path in~$D$.
 We claim that~$(C_1,\dots,C_\tau)$ is a solution to~$I$.
 First observe that~$|C_t|\leq k$ and~$\score_t(C_t)\geq x$.
 Moreover,
 we have that $|\symdif{C_t}{C_{t+1}}|\leq \ell$ ($\geq \ell$ in the revolutionary case)
 for each~$t\in\set{\tau-1}$.
 The claim thus follows.
\lqed
\end{proof}%
}%

\begin{theorem}[\appref{thm:xpregardingk}]
 \label{thm:xpregardingk}
 \cmpv{} and \rmpv{} both admit an~$\O(\tau\cdot m^{2k+1}\cdot n)$-time algorithm and hence are contained in~\XP{} when parameterized by~$k$.
\end{theorem}

\noindent
Since the committee size~$k$ must be at most $m$, 
we obtain the following fixed-parameter tractability result regarding $m$.

\begin{corollary}
\label{cor:FPTwrtm}
 \cmpv{} and \rmpv{} both are solvable in time $2^{\O(m\tlog{m})}\cdot\tau n$.
\end{corollary}

\noindent 
The above two results are quite meaningful for elections with few
candidates, even more if the committee to be chosen is very small. In fact,
small-scale elections seem quite common in practice, which can be observed in
preflib~\citep{MW17-trends}, %
an open-access collection of real-world election data.
It turns out that 37\% and 48\% of
instances stored therein feature, respectively, at most $10$ and
$30$~candidates. Thus, we believe the above algorithm is promising in the light
of real-world applications, especially as our theoretical bounds only regard the
worst-case complexity.

We move on to the next result, providing an algorithm exploiting a small
number~$\tau$ of stages.

\begin{theorem}[\appref{thm:xptau}]
 \label{thm:xptau}
 When parameterized by~$\tau$,
 \cmpv{} and \rmpv{} are contained in~\XP{}.
\end{theorem}

\appendixproof{thm:xptau}{
\begin{proof}
 We describe a dynamic programming algorithm using the key insight
 that there are only $2^\tau$ possible subsets of committees for
 a single candidate to be part of.
 With $\tau$~being constant, we can afford to iterate through all
 these possibilities.
 Even more importantly, when building up all committees simultaneously in $m$~phases, with increasing~$i$,
 we consider in each phase~$i$ only committee members among the first $i$~candidates.
 Herein, we keep track of the sizes, symmetric differences, and scores
 of all $\tau$~committees using only polynomially many table entries.

 We use a boolean dynamic programming table $T[i$, $k_1, \dots, k_\tau$, $d_1,
 \dots, d_{\tau-1}, s_1,\dots,s_\tau]$ with $0 \le i \le m, 0 \le k_t \le k, 0
 \le d_t \le 2k, 0 \le s_t \le n$ to implement the above idea.
 The meaning of the table entries is as follows.
 Let $T[i, k_1, \dots, k_\tau, d_1, \dots, d_{\tau-1}, s_1,\dots,s_\tau]$ be true
 if and only if there are committees~$C_1$, $\dots$, $C_\tau \subseteq \{c_1, \dots, c_i\}$
 with 
 \begin{itemize}
  \item $|C_t|=k_t$ for each~$t \in \set{\tau}$,
  \item $|\symdif{C_t}{C_{t+1}}| = d_t$ for each $t \in \set{\tau-1}$, and
  \item $\score_t(C_t)=s_t$ for each~$t \in \set{\tau}$.
 \end{itemize}
 It can be easily verified (using the definition of the boolean table), that there exists a solution to our problem
 if and only if $T[m,k'_1, \dots, k'_\tau,d'_1, \dots, d'_{\tau-1}, s'_1,\dots,s'_\tau]$ is true 
 for some combination of  $k'_1,\dots,k'_\tau$, $d'_1, \dots, d'_{\tau-1}$, $s'_1,\dots,s'_\tau$
 with $k'_t\le k$ for each $t \in \set{\tau}$,
 $d'_t\le \ell$ (respectively~$d'_t\ge \ell$, in the revolutionary case) for each $t \in \set{\tau-1}$,
 and  $s'_t\ge x$ for each $t \in \set{\tau}$.
 The table is of size~$\O(m \cdot k^{2\tau} \cdot n^{\tau})$.
 It remains to show how to fill the table
 and that each table entry can be computed efficiently.
 
 As initialization, we set $T[0,0,\dots,0]$ to true and all other table entries $T[0,\dots]$ to false.
 For each~$i>0$, 
 we compute all entries $T[i,\dots]$ as follows,
 assuming that the entries $T[i-1,\dots]$ have been computed.
 For each \emph{candidate fingerprint}~$F\subseteq 2^{\set{\tau}}$ we set
 $T[i+1,k_1, \dots, k_\tau$, $d_1, \dots, d_{\tau-1}, s_1,\dots,s_\tau]$ to true if
 $T[i,k'_1, \dots, k'_\tau$, $d'_1, \dots, d'_{\tau-1}, s'_1,\dots,s'_\tau]$ is true where
 for each $t \in \set{\tau}$ we have 
  \begin{align*}
  k'_t&=\begin{cases} 
    k_t-1 &\text{if }t \in F,\\
    k_t &\text{otherwise,}
    \end{cases}
    \\
    d'_t&=\begin{cases} 
        d_t-1 &\text{if }t \in F \text{ xor } t+1 \in F,\\
        d_t &\text{otherwise, and}
       \end{cases}
    \\
    s'_t&=\begin{cases} 
        s_t-|u_t^{-1}(c_{i+1})| &\text{if }t \in F,\\
        s_t &\text{otherwise.}
        \end{cases}
  \end{align*}
We set all other entries $T[i+1,\dots]$ to false.

\paragraph{Running time.}
Computing all table entries $T[i,\dots]$ for some~$i$ costs
$\O(2^\tau)$ time steps for setting the ``true''-entries
and another $\O(k^{2\tau} \cdot n^{\tau})$ time steps for setting the ``false''-entries.
Hence, the whole table can be computed in $\O(m \cdot k^{2\tau} \cdot n^{\tau})$ time.

\paragraph{Correctness.} We show that each computed table entry
$T[i, k_1, \dots, k_\tau$, $d_1, \dots, d_{\tau-1}$, $s_1,\dots,s_\tau]$
is indeed set to true if and only if there are
committees~$C_1$, $\dots$, $C_\tau \subseteq \{c_1, \dots, c_i\}$ with 
$|C_t|=k_t$ for each~$t \in \set{\tau}$,
$|\symdif{C_t}{C_{t+1}}| = d_t$ for each $t \in \set{\tau-1}$, and
$\score_t(C_t)=s_t$ for each~$t \in \set{\tau}$. We give an argument by induction.
This is trivially correct for all table entries with~$i=0$.
Now, assume that for some~$i \in [m-1]$ it holds that
$T[i, k_1, \dots, k_\tau, d_1, \dots, d_{\tau-1}, s_1,\dots,s_\tau]$
is set to true if and only if there are
committees~$C_1$, $\dots$, $C_\tau \subseteq \{c_1, \dots, c_i\}$ with 
$|C_t|=k_t$ for each~$t \in \set{\tau}$,
$|\symdif{C_t}{C_{t+1}}| = d_t$ for each $t \in \set{\tau-1}$, and
$\score_t(C_t)=s_t$ for each~$t \in \set{\tau}$.
We show that also each table entry
$T[i+1, k_1, \dots, k_\tau, d_1, \dots, d_{\tau-1}, s_1,\dots,s_\tau]$
is set to true if and only if there are
committees~$C_1$, $\dots$, $C_\tau \subseteq \{c_1, \dots, c_{i+1}\}$ with 
$|C_t|=k_t$ for each~$t \in \set{\tau}$,
$|\symdif{C_t}{C_{t+1}}| = d_t$ for each $t \in \set{\tau-1}$, and
$\score_t(C_t)=s_t$ for each~$t \in \set{\tau}$.

Let $T[i+1, k^*_1, \dots, k^*_\tau, d^*_1, \dots, d^*_{\tau-1}, s^*_1,\dots,s^*_\tau]$
be the table entry of interest.

First, assume that $T[i+1, k^*_1, \dots, k^*_\tau, d^*_1, \dots, d^*_{\tau-1}, s^*_1,\dots,s^*_\tau]$ is set to true.
Herein, let $F^* \subseteq 2^{\{1,\dots,\tau\}}$ be the fingerprint used when
the table entry was set to true.
This implies that
\[T[i, k'_1, \dots, k'_\tau, d'_1, \dots, d'_{\tau-1}, s'_1,\dots,s'_\tau]\]
was true with
\[k'_t=\begin{cases} k^*_t-1 &\text{if }t \in F^*\\
                    k^*_t &\text{otherwise;}
  \end{cases}\]
\[d'_t=\begin{cases} d^*_t-1 &\text{if }t \in F^* \text{ xor } t+1 \in F^*\\
      d_t &\text{otherwise;}
    \end{cases}\]
   \[s'_t=\begin{cases} s^*_t-|u_t^{-1}(c_{i+1})| &\text{if }t \in F^*\\
                    s^*_t &\text{otherwise.}
		  \end{cases}\]
Thus, there are committees~$C'_1$, $\dots$, $C'_\tau \subseteq \{c_1, \dots, c_i\}$ with 
$|C'_t|=k'_t$ for each~$t \in \set{\tau}$,
$|\symdif{C'_t}{C'_{t+1}}| = d'_t$ for each $t \in \set{\tau-1}$, and
$\score_t(C'_t)=s'_t$ for each~$t \in \set{\tau}$.
It is easy to verify that we obtain the desired committees
$C^*_1$, $\dots$, $C^*_\tau \subseteq \{c_1, \dots, c_{i+1}\}$ via
by
\[C^*_t=\begin{cases} C^*_t\cup\{c_{i+1}\} &\text{if }t \in F^*\\
                        C^*_t &\text{otherwise}
		      \end{cases}\]
and it holds that $|C^*_t|=k^*_t$ for each~$t \in \set{\tau}$,
$|\symdif{C^*_t}{C^*_{t+1}}| = d^*_t$ for each $t \in \set{\tau-1}$, and
$\score_t(C^*_t)=s^*_t$ for each~$t \in \set{\tau}$.

Second, assume that $T[i+1, k^*_1, \dots, k^*_\tau, d^*_1, \dots, d^*_{\tau-1}, s^*_1,\dots,s^*_\tau]$ is set to false.
Assume towards a contradiction that there are
committees~$C^*_1$, $\dots$, $C^*_\tau \subseteq \{c_1, \dots, c_{i+1}\}$ with 
$|C^*_t|=k^*_t$ for each~$t \in \set{\tau}$,
$|\symdif{C^*_t}{C^*_{t+1}}| = d^*_t$ for each $t \in \set{\tau-1}$, and with
$\score_t(C^*_t)=s^*_t$ for each~$t \in \set{\tau}$.
Obviously, there must be committees
$C'_1$, $\dots$, $C'_\tau \subseteq \{c_1, \dots, c_i\}$
via $C'_t=C^*_t \setminus \{c_{i+1}\}$
with $|C'_t|=k'_t$ for each~$t \in \set{\tau}$,
$|\symdif{C'_t}{C'_{t+1}}| = d'_t$ for each $t \in \set{\tau-1}$, and
$\score_t(C'_t)=s'_t$ for each~$t \in \set{\tau}$.
Moreover, by our inductive assumption, we have
\[T[i,k'_1, \dots, k'_\tau, d'_1, \dots, d'_{\tau-1},
s'_1,\dots,s'_\tau]=\textrm{true.}\]
Thus,  our
algorithm would have set $T[i+1, k^*_1, \dots, k^*_\tau, d^*_1, \dots,
d^*_{\tau-1}, s^*_1,\dots,s^*_\tau]$ to true using the fingerprint~$F^*=\{t \mid
c_{i+1} \in C^*_t\}$; a contradiction.
\end{proof}
}

\noindent
The $\xp$ containment shown in \cref{thm:xptau} is surprising because known results for multistage or temporal
(graph)
problems show either
\NP-hardness for constant lifetime or trivial
fixed-parameter tractability for this parameter.
In practice,
the algorithm from~\cref{thm:xptau} could prove useful for
short-term planning, which is inevitable for successful planning.

The last result features the difference between~\rmpv{}
and~\cmpv{} manifesting in the impact on the computational complexity of
the difference~$\ell$ of consecutive committees.

\begin{theorem}[\appref{thm:rmpvXPell}]
 \label{thm:rmpvXPell}
 Every instance~$I=(A,C,U,k,\ell,x)$ of~\rmpv{} with $n$~agents, 
 $m$ candidates,
 and~$\tau$ voting profiles
 can be decided in~$\O(\tau \cdot m^{4\ell+1}\cdot n)$ time.
\end{theorem}

Inspired by a first draft of our paper,
recently
\citet{KRZ21} proved that \rmpv{} is solvable in~$2^{\O(\ell)}\cdot
(n+m+\tau)^{\O(1)}$~time---it is in FPT when parameterized only by~$\ell$.

\appendixproof{thm:rmpvXPell}
{
\noindent
Note that we can assume that~$\ell\leq 2k$.
The crucial observation behind the $\XP$-algorithm comes from the structure of
every solution~$\calC=(C_1,\ldots,C_\tau)$:
Since~$|\symdif{C_i}{C_{i+1}}|\geq \ell$, 
there are~$X\subseteq C_i\setminus C_{i+1}$ 
and $Y\subseteq C_{i+1}\setminus C_i$ such that~$X\cap Y=\emptyset$ and~$|X\cup Y|\geq\ell$.
This allows us to build a directed graph~$D$ as follows (see~\cref{fig:xpell}
for an illustration).
{
\begin{figure*}[t]
   \centering
  \begin{tikzpicture}
  \def\xr{1.0}
  \def\yr{0.75}
  \def\xsh{2.4}
  \def\bxh{3}
  \def\bxw{1.33}
  \def\vdsc{0.8}

  \tikzstyle{xnode}=[circle,fill,scale=0.5,draw]
  \tikzstyle{xedge}=[thick,-];
  \tikzstyle{xarc}=[thick,->,>=latex];
  \tikzstyle{earc}=[lightgray,dashed,thick,->,>=latex];

  \node (s) at (0,0)[xnode,label=-90:{$s$}]{};
  \node (t) at (5*\xr*\xsh,0)[xnode,label=-90:{$t$}]{};

  \node (hl1) at (.85*\xr*\xsh,0.3*\bxh*\yr)[]{$V^{1}$};
  \node at (hl1.north west)[anchor=north west,dotted,rounded corners,minimum width=\xr*\bxw cm,minimum height=\yr*\bxh cm,draw]{};
  \node (v1) at (1*\xr*\xsh,-0.3*\yr)[xnode,label=-90:{$v^{1}_{X,Y}$}]{};
  \node at (v1.north)[anchor=south,yshift=\yr*0.15cm,scale=\vdsc]{$\vdots$};
  \node at (v1.south)[anchor=north,yshift=-\yr*0.6cm,scale=\vdsc]{$\vdots$};

  \node (hl2) at (1.85*\xr*\xsh,0.3*\bxh*\yr)[]{$V^{2}$};
  \node at (hl2.north west)[anchor=north west,dotted,rounded corners,minimum
  width=\xr*\bxw cm,minimum height=\yr*\bxh cm,draw]{};
  \node (v2) at (2*\xr*\xsh,-0.2*\yr)[xnode,label=-90:{$v^{2}_{X',Y'}$}]{};
  \node at (v2.north)[anchor=south,yshift=\yr*0.15cm,scale=\vdsc]{$\vdots$};
  \node at (v2.south)[anchor=north,yshift=-\yr*0.6cm,scale=\vdsc]{$\vdots$};

  \node (br) at (3*\xr*\xsh,0)[scale=1.33]{$\cdots$};

  \node (hl4) at (3.85*\xr*\xsh,0.3*\bxh*\yr)[]{$V^{\tau-1}$};
  \node at (hl4.north west)[anchor=north west,dotted,rounded corners,minimum
  width=\xr*\bxw cm,minimum height=\yr*\bxh cm,draw]{};
  \node (v4) at (4*\xr*\xsh,-0.4*\yr)[xnode,label=-90:{$v^{\tau-1}_{X'',Y''}$}]{};
  \node at (v4.north)[anchor=south,yshift=\yr*0.15cm,scale=\vdsc]{$\vdots$};
  \node at (v4.south)[anchor=north,yshift=-\yr*0.6cm,scale=\vdsc]{$\vdots$};

  \draw[xarc] (s) to node(a)[midway]{}(v1);
  \node (txt) at (-0.35*\xsh,3.*\yr)[align=left, anchor=north
  west,font=\small,rounded corners,fill=lightgray!15!white]{$\iff$ $\exists$ $k$-sized $C'\subseteq C$ containing $X$ being  disjoint 
  from $Y$ with $\score_1(C')\geq x$};
  \draw[earc] (txt.south west) to [out=-135,in=90](a);

  \draw[xarc] (v1) to node(a)[midway]{}(v2);
  \node (txt) at (-0.875*\xsh,-2.*\yr)[align=left, anchor=north
  west,font=\small,rounded corners,fill=lightgray!15!white]{\parbox{15cm}{$\iff$ $X\cap X'=Y\cap Y'=\emptyset$ and
  $\exists$ $k$-sized $C'\subseteq C$ containing  $X'\cup Y$ being
  disjoint from $X\cup Y'$ with $\score_2(C')\geq x$}};
  \draw[earc] (txt) to [out=135,in=-90](a);

  \draw[xarc] (v2) to node(a)[midway]{}(br);
  \draw[xarc] (br) to node(a)[midway]{}(v4);

  \draw[xarc] (v4) to node(a)[midway]{}(t);
  \node (txt) at (0.5*\xsh,2.125*\yr)[align=left, anchor=north west,font=\small,rounded corners,fill=lightgray!15!white]{$\iff$ $\exists$ $k$-sized $C'\subseteq C$ containing $Y''$ being disjoint 
  from $X''$ with $\score_\tau(C')\geq x$};
  \draw[earc] (txt.south east) to [out=-45,in=90](a);
  \end{tikzpicture}
  \caption{Illustration of an in-out graph from~\cref{def:inoutgraph}.}
  \label{fig:xpell}
\end{figure*}
}

\begin{definition}
 \label{def:inoutgraph}
 Given an instance~$I=(A,C,U,k,\ell,x)$ of~\rmpv{},
 the \emph{in-out} graph of~$I$ is a directed graph~$D_I$ with vertex set 
 $V =V^1\cup \dots\cup V^{\tau-1}\cup\{s,t\}$ where
  $V^i = \{v^i_{X,Y}\mid X,Y\subseteq C,\, |X|,|Y|\leq \ell, X\cap Y=\emptyset,|X\cup Y|\geq \ell\}$ for every~$i\in\set{\tau-1}$,
 containing the arcs~$(s,v^1_{X,Y})$ if and only if there is a~$k$-sized committee at time step~$1$ containing~$X$ being disjoint from~$Y$ with score at least~$x$,
 the arcs~$(t,v^{\tau-1}_{X,Y})$ if and only if there is a~$k$-sized committee at time step~$\tau$ containing~$Y$ being disjoint from~$X$ with score at least~$x$,
 and an arc~$(v^i_{X,Y},v^j_{X',Y'})$ if and only if
 \begin{inparaenum}[(i)]
  \item $j=i+1$,~$X\cap X'=Y'\cap Y=\emptyset$, and 
  \item there is a~$k$-sized committee at time step~$i+1$ containing~$Y\cup X'$ being disjoint from~$X\cup Y'$ with score at least~$x$.
 \end{inparaenum}
\end{definition}

\noindent
With an \XP-running time (regarding~$\ell$),
we can compute the in-out graph for every given instance.
\begin{lemma}%
 \label{lem:inoutgraphrunningtime}
 Given an 
 instance
 of~\rmpv{} with $n$~agents, 
 $m$ candidates,
 and~$\tau$ voting profiles,
 the in-out graph~$D_I$ of~$I$ can be computed 
 in~$\O(\tau \cdot m^{4\ell+1}n)$~time.
\end{lemma}

{
\begin{proof}
 We can compute each~$V^i$ in~$\O(m^{2\ell})$ time by brute forcing every pair of candidate subsets each of size at most~$\ell$.
 Hence,
 we can compute~$V$ in~$\O(\tau\cdot m^{2\ell})$ time.
 The arcs incident with~$s$ and~$t$ can be computed in~$\O(n+m)$ time,
 same for conditions~(i) and~(ii),
 resulting in an overall running time in~$\O(\tau \cdot m^{4\ell+1}n)$.
\lqed
\end{proof}
}

\noindent
We prove that deciding an instance of~\rmpv{} can be done through deciding whether there is an~$s$-$t$ path in the instance's in-out graph.

\begin{lemma}%
 \label{lem:pathinoutgraph}
 Let~$I$ be an instance of~\rmpv{} and~$D_I$ its in-out graph.
 Then, 
 there is an~$s$-$t$~path in~$D_I$ if and only if $I$~is a \yes-instance.
\end{lemma}

{
\begin{proof}
  Let $I=(A,C,U,k,\ell,x)$ be an instance of~\rmpv{},
  and~$D_I$ the in-out graph for~$I$.
  
  \RD{}
  It follows directly from our construction of~$D_I$ that 
  each~$s$-$t$~path in~$D_I$ is of the form~$P=(s,v_{X_1,Y_1}^1,\dots,v_{X_{\tau-1},Y_{\tau-1}}^{\tau-1},t)$.
  Since arc~$(s,v_{X_1,Y_1}^1)$ exists,
  by construction there is a~$k$-size committee~$C_1$ containing~$X$ but disjoint from~$Y$ of score at least~$x$ (similar for~$C_\tau$).
  Since the arc~$(v_{X_{i},Y_{i}}^{i},v_{X_{i+1},Y_{i+1}}^{i+1})$ exists,
  at time step~$i+1$ there is a~$k$-size committee~$C_{i+1}$ containing~$Y_i\cup X_{i+1}$ but disjoint from~$X_i\cup Y_{i+1}$ of score at least~$x$.
  Observe that~$\symdif{C_i}{C_{i+1}}\supseteq X_i\cup Y_i$ and thus, 
  of size at least~$\ell$.
  Then we find in each time step a committee of score~$x$,
  and consecutive committees differ in at least~$\ell$ elements.

  \LD{}
  Assume there is a solution~$\calC=(C_1,\ldots,C_\tau)$.
  For every~$i\in\set{\tau-1}$,
  let~$X_i\subseteq C_i\setminus C_{i+1}$,
  $Y_i\subseteq C_{i+1}\setminus C_i$.
  Note that~$C_1$ is a committee containing~$X$ but disjoint from~$Y$ of score at least~$x$ in time step one,
  and thus arc~$(s,v^1_{X_1,Y_1})$ exists (analogously for~$C_\tau$ and arc~$(v^{\tau-1}_{X_{\tau-1},Y_{\tau-1}},t)$).
  We claim that there is an arc from~$v^i_{X_i,Y_i}$ to~$v^{i+1}_{X_{i+1},Y_{i+1}}$ for every~$i\in\set{\tau-2}$.
  Note that since~$X_i\cap C_{i+1}=\emptyset$, we have~$X_i\cap X_{i+1}=\emptyset$ (analogously for~$Y$).
  Moreover,
  $C_{i+1}$ is a $k$-size committee with score at least~$x$ (since~$\calC$ is a solution),
  and it contains~$Y_i\cup X_{i+1}$, and is disjoint from~$X_i\cup Y_{i+1}$.
  Hence, 
  the arc exists.
  It follows that~$P=(s,v^1_{X_1,Y_1},\dots,v^{\tau-1}_{X_{\tau-1},Y_{\tau-1}},t)$ forms an~$s$-$t$ path in~$D_I$.
\lqed
\end{proof}
}

\noindent
Given~\cref{lem:inoutgraphrunningtime,lem:pathinoutgraph},
we are set to prove~\cref{thm:rmpvXPell}.

{
\begin{proof}[Proof of \cref{thm:rmpvXPell}]
 Let $I=(A,C,U,k,\ell,x)$ be an instance of~\rmpv{} with~$n$ agents, 
 $m$ candidates,
 and~$\tau$ voting profiles.
 First, 
 construct the in-out graph~$D_I$ of~$I$ in~$\O(\tau \cdot \ell^4 m^{4\ell+1}n)$ time (due to~\cref{lem:inoutgraphrunningtime}).
 Next, 
 in time linear in the size of~$D_I$
 check for an~$s$-$t$ path in~$D_I$.
 Due to~\cref{lem:pathinoutgraph},
 if an~$s$-$t$ path is found,
 then report that~$I$ is a \yes{}-instance,
 and otherwise,
 if no such~$s$-$t$ path is found,
 then report that~$I$ is a \no-instance,
\lqed
\end{proof}
}

\noindent
We remark that the proof of~\cref{lem:pathinoutgraph} contains the description of how to make the algorithm constructive 
(that is, if one requires to return a solution).
}

  \noindent 

\section[Provably Effective Efficient Data Reduction]{Provably Effective Efficient Data Reduction}
\label{sec:preprocessing}
\appendixsection{sec:preprocessing}

In this section we discuss efficient procedures that can be used to preprocess
an instance in order to simplify its further processing. 
Preprocessing can be quite effective even for \NP-hard problems~\citep{Weihe98}.
In terms of parameterized complexity,
preprocessing with guarantee is called \emph{problem
kernelization}:
For a
parameterized problem~$L$,
it is a polynomial-time algorithm that maps any
instance~$(x,p)\in\Sigma^*\times \N_0$ of~$L$ to an equivalent
instance~$(x',p')$ of~$L$ (a \emph{problem kernel}) such that~$|x'|+p'\leq
f(p)$ for some function~$f$ only depending on the parameter~$p$.

Preferably, we want~$f$ to be some polynomial, in which case we
call the problem kernelization \emph{polynomial}. Polynomial problem
kernelizations serve as efficient and provably effective data reductions, which
intuitively ``cut off'' ``obvious''
parts of an instance. After such preprocessing, a final algorithm fed with the
obtained kernel can perform significantly better compared to the original
instance.

\subsection{Number of Candidates and Stages Combined}

We first consider a kernelization regarding the numbers~$m$ of candidates
and~$\tau$ of stages. 

\begin{theorem}[\appref{thm:PKwrtmtau}]
\label{thm:PKwrtmtau}
 \cmpv{} and \rmpv{} %
 admit problem kernels of size polynomial in~$m+\tau$.
\end{theorem}

\noindent
The proof of~\cref{thm:PKwrtmtau}
uses weighted versions of our problems
(called~\Wcmpv). 
Each weighted version takes, 
for each stage, 
a vector of size~$m$ in which each entry $i \in \{1 ,\dots, m\}$ corresponds to the number of approvals that candidate~$i$ gets in a
given stage 
(so, the number of agents upper-bounds the sum
of all entries of the vector).
\appendixproof{thm:PKwrtmtau}
{

\noindent
The reduction from~\cmpv{} to \Wcmpv{} is obvious (the sum of weights equals~$n$).

\begin{observation}
 \label{obs:cmpvtowmcpv}
 There is a polynomial-time many-one reduction from \cmpv{} to \Wcmpv{} with $\tau$ weight vectors from~$\set[0]{n}^m$.
\end{observation}

{
\begin{proof}
 For each~$t\in\set{\tau}$,
 set~$w^t_i$ equal to the number of approvals of candidate~$c_i$ in the~$t$th utility function,
 that is,
 $w^t_i=|u_t^{-1}(c_i)|$.
\lqed
\end{proof}
}

{

We can shrink the weights of the weighted version of each of our problems the following result due to \citet{FrankT87}.

\begin{proposition}[{\cite[Section~3]{FrankT87}}]
  \label[proposition]{thm:FrankTardos}
  There is an algorithm that,
  on input~$w\in\Q^d$ and integer~$N$,
  computes in polynomial time
  a vector~$\wmod{w}\in \Z^d$ with
  \begin{compactenum}[(i)]
   \item $\norm{\wmod{w}}{\infty}\leq 2^{4d^3}N^{d(d+2)}$ such that
   \item $\sign(w^\top b)=\sign(\wmod{w}^\top b)$
  for all~$b\in\Z^d$ with~$\norm{b}{1}\leq N-1$.
  \end{compactenum}
\end{proposition}

\cref{thm:FrankTardos} now gives the following.
}

\begin{lemma}
 \label{lem:FToncmpv}
 There is an algorithm that, 
 given an instance~$I=(W,k,\ell,x)$ of \Wcmpv{},
 computes in polynomial time an instance $I'\ceq (W'=(\wmod{w}^1,\dots,\wmod{w}^\tau),k,\ell,\wmod{x})$ such that
 \begin{inparaenum}[(i)]
  \item $\norm{\wmod{w}^t}{\infty},\allowbreak |\wmod{x}|\in 2^{\O(m^3\tau^3)}$ for all~$t\in\set{\tau}$, and
  \item $\calC$ is a solution to~$I$ if and only if it is a solution to~$I'$.
 \end{inparaenum}
\end{lemma}

{
\begin{proof}
 Let~$\omega\in\N_0^{m\cdot\tau}$ the concatenation of the weight vectors in~$W$,
 and let~$w\ceq (\omega,x)$.
 Apply \cref{thm:FrankTardos} to~$w$ with~$N=k+2$ (note that~$d=m\cdot\tau+1$)
 to obtain a vector~$\wmod{w}=(\wmod{\omega},\wmod{x})$.
 Property (i) holds true by~\cref{thm:FrankTardos}(i).
 Let~$\calC=(C^1,\ldots,C^\tau)$ be a sequence of subsets of~$\set{m}$.
 Let~$b^t\in\{0,1\}^m$ be the vector associated with~$C^t$ 
 (that is,~$b^t_i=1$ if and only if~$i\in C^t$).
 Then, 
 by \cref{thm:FrankTardos}(i) (since~$\norm{b^t}{1}\leq k$),
 $(b^t,-1)^\top (w^t,x)\geq0$ if and only if~$(b^t,-1)^\top (\wmod{w}^t,\wmod{x})\geq0$.
 Hence,
 $\calC$ is a solution to~$I$ if and only if it is to~$I'$,
 proving property~$(ii)$.
\lqed
\end{proof}
}

\noindent
It is not hard to see that~\Wcmpv{} is \NP-complete.
We are set to prove our main result.

\begin{proof}[Proof of~\cref{thm:PKwrtmtau}]
 Let~$I=(A,C,k,\ell,x)$ be an instance of \cmpv{}.
 Compute in polynomial time an instance~$J=(W,k,\ell,x)$ of \Wcmpv{} being equivalent to~$I$ 
 (\cref{obs:cmpvtowmcpv}).
 Next apply~\cref{lem:FToncmpv} to~$J$ to obtain an equivalent instance~$J'$ of~\Wcmpv{} of encoding length in~$\O((m\tau)^4)$.
 Finally,
 reduce~$J'$ back to an instance~$I'$ of \cmpv{} in polynomial time.
 Hence,
 the encoding length of~$I'$ is in~$(m\tau)^{\O(1)}$,
 proving~$I'$ to be a polynomial problem kernel.
 For~\rmpv{},
 the proof works analogously.
\lqed
\end{proof}

}%
Roughly put, 
in the proof of~\cref{thm:PKwrtmtau} 
one takes the original instance of~\cmpv{}, 
translates it into an instance of~\Wcmpv{}
(by simply stage-wise computing the approvals of each candidate),
compresses its weights
(via a result
by~\citet{FrankT87}), and finally translates it back into a new instance
of~\cmpv{}, which then forms the desired kernel.

\subsection{Number of Agents and Stages Combined}

\cref{thm:PKwrtmtau} has also an appealing intuitive interpretation. 
Namely,
when there are few candidates and few stages, then the instance cannot be too large. 
In fact,
a complementary intuition is also true:
if there are few agents and few
stages, then there cannot be too many meaningful candidates. 
Formally:

\begin{theorem}%
 \label{thm:pkntau}
 \cmpv{} and \rmpv{} %
 admit problem kernels of size polynomial in~$n+\tau$.
\end{theorem}

\noindent
\cref{thm:pkntau}
makes use of 
several data reduction rules 
which explicitly show how to prune the unnecessary candidates.
To this end, 
recall that there are at most~$n\cdot \tau$ approvals in any instance.
Hence, we have the following.

\begin{observation}
 \label{obs:nonapproved}
 There are at least~$\max\{0,m-n\cdot\tau\}$ candidates which are never approved.
\end{observation}

\noindent
Upon~\cref{obs:nonapproved},
we will next discuss deleting candidates which are never approved,
in order to upper-bound the number~$m$ of candidates 
by some polynomial 
in~$n+\tau$.
Then, we can apply
\cref{thm:PKwrtmtau} to
obtain the polynomial-size problem kernels.
We treat 
\cmpv{} and~\rmpv{} separately.

\subsubsection{Deleting Nowhere-approved Candidates for~\cmpv{}}

For~\cmpv{},
\cref{obs:nonapproved} allows us to reduce any instance 
to an equivalent instance with~$m\leq n\cdot \tau$.

\toappendix{
\subsection{Proof of \cref{thm:pkntau}}
\label{proof:thm:pkntau}

\begin{proposition}
 \label{prop:kernelcmpvntau}
 There is a polynomial-time algorithm that for an
 instance~$I=(A,C,U,k,\ell,x)$ of~\cmpv{}
 computes an equivalent instance~$I'=(A,C',U,k,\ell,x)$ with~$|C'|\leq |A|\cdot |U|$.
\end{proposition}

\noindent
To obtain~\cref{prop:kernelcmpvntau}, 
we employ \cref{rrule:nonappcandidates}.
}

\begin{rrule}%
 \label{rrule:nonappcandidates}
 For an instance of~\cmpv{}, if~$m>n\tau$, delete a candidate which is never
 approved.
\end{rrule}

\noindent
Intuitively,
\cref{rrule:nonappcandidates} is correct because selecting a candidate which is never approved into a committee at some stage is not beneficial:
it only increases the symmetric difference at the respective stage but not the committee's score.%

\toappendix
{%
\begin{proof}[Proof of~\cref{rrule:nonappcandidates}]%
 Let~$I'$ be the instance obtained from~$I$ by \cref{rrule:nonappcandidates} deleting~$z\in C$ which is never approved (its existence follows from~\cref{obs:nonapproved}).
 If~$\calC'$ is a solution to~$I'$, then it is also to~$I$.
 Conversely,
 if there is a solution that does not contain~$z$,
 then this is a solution to~$I'$.
 So,
 assume that every solution contains~$z$.
 
 Let~$\calC=(C_1,\ldots,C_\tau)$ be a solution to~$I$ such that the first appearance of~$z$ in the sequence is the latest.
 Let~$z\in C_{t_1}$, 
 where~$t_1$ is the smallest index with this property.
 Let~$t_2>t_1$ be the largest index such that~$z\in C_{t_2}$ for all~$t_1\leq t\leq t_2$.
 We claim that deleting~$z$ from~$C_{t}$ for~$t_1\leq t\leq t_2$ constructs another solution with the first appearance of~$z$ being later,
 contradicting our choice of~$\calC$.
 It holds true that for all~$t_1\leq t\leq t_2$,
 we have~$|C_t\setminus\{z\}|=|C_t|-1<k$ and,
 since~$z$ is never approved,
 we have that~$\score_t(C_t\setminus\{z\})\geq x$.
 Moreover, 
 for each~$|\symdif{C_t}{C_{t+1}}|\leq\ell$ for each~$t_1\leq t<t_2$.
 Finally,
 we have that~$|\symdif{C_{t_1-1}}{(C_{t_1}\setminus\{z\})}|< |\symdif{C_{t_1-1}}{C_{t_1}}|\leq \ell$,
 since~$z\not\in C_{t_1-1}$.
 Similarly,
 $|\symdif{(C_{t_2}\setminus\{z\})}{C_{t_2+1}}|<|\symdif{C_{t_2}}{C_{t_2+1}}|\leq \ell$,
 since~$z\not\in C_{t_2+1}$.
 Hence, 
 we constructed a solution where the first appearance of~$z$ is later than for~$\calC$,
 contradicting the choice of~$\calC$.
 It follows that there is a solution not containing~$z$,
 witnessing that~$I$ is a \yes-instance.
 Thus,
 the correctness of the rule follows.
 
 The claimed running time is immediate:
 in linear time,
 identify $m-n\tau$ candidates that are never approved and delete them from the candidate set.
\lqed
\end{proof}%

We can safely apply~\cref{rrule:nonappcandidates} exhaustively in polynomial-time,
hence we get~\cref{prop:kernelcmpvntau}.
}%

\subsubsection{Deleting Nowhere-approved Candidates for~\rmpv{}}

For~\rmpv{}, a similar, albeit technically much more involved, approach
applies. Complications stem from the fact that it is not so clear that deleting
an unapproved candidate is safe. 

\begin{rrule}[\appref{rrule:rmpv-m-larger-than-max-nk-tau-O}]
 \label{rrule:rmpv-m-larger-than-max-nk-tau-O}
 For an instance of~\rmpv{},
 if~$m > \max\{n,k\}\cdot\tau$, 
 then delete a candidate that is never approved.
\end{rrule}

\appendixproof{rrule:rmpv-m-larger-than-max-nk-tau-O}
{
  \begin{proof}[Proof of~\cref{rrule:rmpv-m-larger-than-max-nk-tau-O}]
    Let~$I=(A,C,U,k,\ell,x)$ be an instance of~\rmpv{} with~$m > \max\{n,k\}\cdot\tau$, 
    and let~$I'=(A,C',U,k,\ell,x)$ be the instance of~\rmpv{} obtained from applying the reduction rule to~$I$, 
    where~$C'=C\setminus\{z\}$ 
    ($z$ is a candidate that is never approved, 
    which exists since~$m>n\tau$).
    We claim that~$I$ is a \yes-instance
    if and only if
    $I'$~is a \yes-instance.

    \LD{}
    Immediate.

    \RD{}
    Let~$I$ be a \yes-instance.
    We claim that there is a solution~$\calC'=(C_1',\ldots,C_\tau')$ such that~$z\not\in \bigcup_{t=1}^\tau C_t'$.
    Suppose not,
    and let~$\calC=(C_1,\ldots,C_\tau)$ be a solution to~$I$ which contains~$z$.
    Since~$|C_t|\leq k$,
    there is~$y\in C$ such that~$y\not\in \bigcup_{t=1}^\tau C_t$,
    as~$m> k\tau$.
    We claim that~$\wmod{\calC}=(\wmod{C}_1,\ldots,\wmod{C}_\tau)$ with 
    \[ \wmod{C}_t = \begin{cases} C_t,& z\not\in C_t\\ (C_t\setminus\{z\})\cup\{y\},& \text{otherwise,}\end{cases}\]
    is a solution to~$I$ 
    (roughly, 
    $\wmod{\calC}$ is obtained from~$\calC$ by replacing each occurrence of~$z$ by~$y$).
    Clearly,
    $|\wmod{C}_t|\leq k$, $\score_t(\wmod{C}_t)\geq x$,
    and even~$|\symdif{\wmod{C}_t}{\wmod{C}_{t+1}}|\geq \ell$,
    since whenever~$z\in (\symdif{C_t}{C_{t+1}})$, 
    we now have~$y\in(\symdif{\wmod{C}_t}{\wmod{C}_{t+1}})$.
    It follows that~$\wmod{\calC}$ is a solution not containing~$z$.
    Hence,
    $\calC'$ exists and is a solution to~$I'$.
  \lqed
  \end{proof}
}

Indeed, this can decrease the symmetric
difference between some consecutive committees in a potential solution
below~$\ell$. What follows is that one should rather remove unapproved
candidates in pairs. 
Proving that this approach works requires a tedious
case-distinction based analysis, 
so we defer the proof to the appendix.

\begin{rrule}%
 \label{rrule:rmpv-m-larger-than-max-nk-tau}
 For an instance of~\rmpv{} with~$\ell\geq 2$,
 if~$m\geq k > n\tau+3$, 
 then
 delete two candidates which are never approved, 
 decrease the required committee size~$k$ by~one, 
 and decrease the required symmetric difference~$\ell$ by~two.
\end{rrule}

Recall that \rmpv{} is polynomial-time solvable for any constant~$\ell$
(in particular if~$\ell<2$)
\citep{KRZ21}.
\newcommand{\scA}{c}%
\newcommand{\scB}{c'}%
\newcommand{\scC}{d}%
For the proof of \cref{rrule:rmpv-m-larger-than-max-nk-tau},
we use the following.
  
\begin{lemma}[\appref{lem:rmsntv-k-between-m-n}]
  \label{lem:rmsntv-k-between-m-n}
  Let $(A,C,U,k,\ell,x)$ be a \yes-instance of~\rmpv{} with 
  $m\geq k > n\tau+3$.
  Then for every nowhere approved candidates~$\scA,\scB\in C$
  there is a solution~$\calC=(C_1,\dots,C_\tau)$ such that
  removing~$\scA,\scB$,
  i.e.,
  $\calC'=(C_1',\dots,C_\tau')$ with~$C_i'=C_i\setminus \{\scA,\scB\}$
  yields committees each of size at most~$k-1$.
\end{lemma}

\appendixproof{lem:rmsntv-k-between-m-n}
{
\begin{proof}
  Before we present an argument, we introduce an operator that is fundamental to
  our reasoning. Consider some solution $\calC=(C_1, C_2, \ldots,C_\tau)$ to an
  of~\rmpv{} over a given set of candidates~$C$. Then, for each pair of distinct
  candidates~$c$ and~$d$, we let $\switch(\calC, i, c, d)$ to output a new
  solution\ \ $\calC^* = (C_1, \ldots, C_{i-1}, C^*_{i}, \ldots, C^*_\tau)$,
  where~$C^*_j$, $j \in \{i, \ldots, \tau\}$, is obtained from~$C_j$ by switching
  each occurence of candidate~$c$ with candidate~$d$ and the other way round.
  Formally, for some~$C_j$, $j \in \{i, \ldots, \tau\}$, the formula describing
  the corresponding committee~$C^*_j$ is
  $$
    C_j^* \ceq 
    \begin{cases} 
      C_j ,& \text{if $C_j\cap \{\scA,\scC\}=\emptyset$},\\
      C_j ,& \text{if $C_j\supseteq \{\scA,\scC\}$},\\
      C_j\setminus \{\scA\}\cup\{\scC\} ,& \text{if $C_j\cap \{\scA,\scC\}=\{\scA\}$},\\
      C_j\setminus \{\scC\}\cup\{\scA\} ,& \text{if $C_j\cap \{\scA,\scC\}=\{\scC\}$}.
    \end{cases}
  $$
  Notably, for each pair of neighbors $(C_j, C_{j+1})$, $j \in \{i, \ldots,
  \tau-1\}$, it holds that~$|\symdif{C_j^*}{C_{j+1}^*}|=|\symdif{C_j}{C_{j+1}}|$.
  This can be easily verified, for example, by analyzing each of the
  $10$~possible configurations of the conditions from the above-formula can be
  met by consecutive committees~$C_j$ and~$C_{j+1}$, $j \in \{i, \ldots,
  \tau-1\}$ (while counting the configurations, we utilized the symmetry of
  the symmetric difference). Importantly, if $c$ and~$d$ are both nowhere
  approved, then the scores of $\calC$ and $\switch(\calC, i, c, d)$ are trivially the
  same for every value of~$i \in \{1, \ldots, \tau\}$. This \emph{preservation} of
  the values of the symmetric differences and the scores is the crucial feature
  of the switch operator.

  We now prove the lemma by contradiction. To this end, assume that the lemma is
  false. This means that there exist two distinct, nowhere approved candidates
  $\{\scA,\scB\} \subseteq C$ such that for every solution
  $\calC=(C_1,\dots,C_\tau)$ we have that $\calC'=(C_1',\dots,C_\tau')$
  with~$C_j'=C_j\setminus \{\scA,\scB\}$, $j \in \{1, \ldots, \tau\}$, there is
  the smallest value~$i$ such that~$|C_i'|=k$.  
  Let~$\calC$ be one of these solutions such that this index~$i$ is maximum.
  Note that since $|C_i|=|C_i'|=k \geq n\tau+4$, we have that~$C_i$ contains a nowhere
  approved candidate~$\scC$ different than~$\scA$ and~$\scB$.

  Using a series of procedures which depend on the value of~$i$, we
  construct a new solution~$\calC^*$ using the switch operator. Then, we show
  that~$\calC^*$ contradicts our initial assumption thus proving the lemma.

  In all subsequent cases, we get~$\calC^*$ by applying the switch operator on
  the solution~$\calC$ over two candidates---$\scC$ and either~$\scA$
  or~$\scB$---starting from some committee~$j \geq 1$. Because the switch
  operator does not change committees before the $j$-th committe, it is clear
  that all of the symmetric differences and the scores of committees $C_1$ to
  $C_{j-1}$ in~$\calC^*$ remain the same as the respective ones in~$\calC$. Due
  to the preservation feature of the switch operator, the symmetric differences
  and the scores (recall that all of~$\scA$, $\scB$, and~$\scC$ are nowehere
  approved) remain the same for committees~$(C^*_j, \ldots, C^*_\tau)$ and
  their corresponding committees from~$\calC$. Hence, in every case we only
  focus on showing that applying the switch operator does not make the size
  of~\symdif{C_{j-1}}{C^*_{j}} smaller than that of~\symdif{C_{j-1}}{C_{j}}.
  By this we show that~$\calC^*$ is indeed a correct solution to the original
  instance of~\rmpv{}.

  \xcase{1}{$i=1$}
  We let $\calC^* = \switch(\calC, 1, \scA, \scC)$. The correctness of~$\calC^*$
  comes from the fact that $i=1$, and thus~$\symdif{C_{i-1}}{C_{i}}$ does not
  exist.

  \xcase{2}{$i>1$ and $C_{i-1}\not\supseteq\{\scA,\scB\}$}
    We distinguish three further cases regarding the intersection of~$C_{i-1}$
    with~$\{\scA,\scB\}$.

  \xsubcase{2a}{$C_{i-1}\cap\{\scA,\scB\}=\emptyset$}
    We let $\calC^* = \switch(\calC, i, \scA, \scC)$ and we have that
    $|\symdif{C_{i-1}}{C_{i}^*}| \geq
    |\symdif{C_{i-1}}{C_{i}}|$. Indeed, since $\scA \notin C_{i-1}$ as well as
    $\scA \notin C_i$ (by the assumption on~$i$), then introducing
    candidate~$\scA$ to~$C_{i}$ and thus obtaining~$C^*_i$
    contributes to the increase of the symmetric difference by~$1$.
    On the other hand, removing~$d$ from~$C_{i}$ leading to
    getting~$C^*_1$ can either decrease the symmetric difference by
    one (if~$d \not\in C_{i-1}$) or not influence the symmetric
    difference at all. So, we obtain the desired value
    of~$|\symdif{C_{i-1}}{C_{i}^*}|$.

  \xsubcase{2b}{$C_{i-1}\cap\{\scA,\scB\}=\{\scB\}$}
    Once more, as in Case~2a, we let $\calC^* = \switch(\calC, i, \scA, \scC)$.
    The presence of~$\scB \in C_{i-1}$ does not affect the argument for
    Case~2a, so it carries over to this case.

  \xsubcase{2c}{$C_{i-1}\cap\{\scA,\scB\}=\{\scA\}$}
    We let $\calC^* = \switch(\calC, i, \scB, \scC)$. This case is symmetric
    to~Case~2b, so the argument from the latter carries on after renaming (in
    the argument) candidate~$\scA$ to~$\scB$.

\xcase{3}{$i=2$ and $C_{i-1} \supseteq\{\scA,\scB\}$}
    We let $\calC^* = \switch(\calC, i-1, \scA, \scC)$. Note that in this case
    $i-1 = 1$, so the corectness of~$\calC^*$ is trivial.

\xcase{4}{$i \geq 3$ and $C_{i-1} \supseteq\{\scA, \scB, \scC\}$}
    We let $\calC^* = \switch(\calC, i-1, \scA, \scC)$. In this case, $C^*_{i-1}
    \ceq C_{i-1}$. So it holds that~$\symdif{C_{i-2}}{C_{i-1}^*} =
    \symdif{C_{i-2}}{C_{i-1}^*}$, which proves the correctness of~$\calC^*$.
    
\xcase{5}{$i \geq 3$ and $C_{i-1} \supseteq\{\scA,\scB\}$ and $\scC \not\in
C_{i-1}$}
  We further subdivide this case into multiple possible situations depending on the
  elements of~$C_{i-2}$.

  \xsubcase{5a}{$\scA\in C_{i-2}$ or~$\{\scA,\scB\}\subseteq C_{i-2}$}
    We let $\calC^* = \switch(\calC, i-1, \scA, \scC)$. Note
    that~$\scA\not\in\symdif{C_{i-2}}{C_{i-1}}$ but in~$\calC^*$ (after the switch) we have
    that~$\scA\in\symdif{C_{i-2}}{C_{i-1}^*}$. Similarly like in Case~2a,
    removing candidate~$d$ as a result of the switch operation cannot decrease
    the symmetric difference by more than one. So it follows that
    $|\symdif{C_{i-2}}{C_{i-1}^*}| \geq |\symdif{C_{i-2}}{C_{i-1}}|$.
    
  \xsubcase{5b}{$C_{i-2} \cap \{\scA,\scB\} = \{\scB\}$}
    We let $\calC^* = \switch(\calC, i-1, \scB, \scC)$. This case is symmetric
    to Case~4a (specifically, to its first condidtion), so the argument
    carries over via relabeling $\scA$ and $\scB$.

  \xsubcase{5c}{$\{\scA,\scB\}\cap C_{i-2}=\emptyset$}
    We differentiate our argument depending on whether~$C_{i-2}$ contains~$\scC$.

    \xsubsubcase{5c}{i}{$\scC \not\in C_{i-2}$}
      We let $\calC^* = \switch(\calC, i-1, \scA,
      \scC)$. 
    Because of the condition of Case~5c, we have that~$\scA
    \in \symdif{C_{i-2}}{C_{i-1}}$. Due to the condition
    in Case~5, we also know that $\scC \not\in C_{i-1}$. Then $C^*_{i-1}
      \ceq C_{i-1} \setminus \{\scA\} \cup \{\scC\}$. So we get that~$\scA \not\in
      \symdif{C_{i-2}}{C_{i-1}^*}$ but~$\scC \in
      \symdif{C_{i-2}}{C_{i-1}^*}$. Since (originally) in this case~$\scC
      \not\in \symdif{C_{i-2}}{C_{i-1}}$, then we get
      $|\symdif{C_{i-2}}{C_{i-1}^*}| = |\symdif{C_{i-2}}{C_{i-1}}|$.

    \xsubsubcase{5c}{ii}{$\scC \in C_{i-2}$}
      
      \def\sidx{b}
      We start with formalizing the situation in which a pair~$\{\scA, \scB\}$~of
      candidates alternate with candidate~$\scC$ in a series of consecutive
      committees. We capture the possibility that such a series can either start with
      a committee containing~$\scC$ or the one containing~$\{\scA, \scB\}$.
      Consider some positive integer~$\sidx{}<i$. We say that
      committees~$C_{\sidx{}}$ up to~$C_{i}$ are \emph{\scC-first alternating
      sequence} 
      if
      it holds that~$\scC\in C_{\sidx{}} \setminus
      C_{\sidx{}+1}$, 
      $\{\scA,\scB\}\subseteq C_{\sidx{} +1}\setminus C_{\sidx{}}$,
      and $\{\scA,\scB,\scC\}\subseteq \symdif{C_{\sidx{}+j-1}}{C_{\sidx{}+j}}$
      for each~$j \in \{1, \ldots, i -\sidx{}\}$.
      Analogously, committees~$C_{\sidx{}}$ up to~$C_{i}$ are \emph{$\{\scA,
      \scB\}$-first alternating sequence} 
      if 
      it holds that~$\{\scA,\scB\}\subseteq C_{\sidx{}} \setminus
      C_{\sidx{}+1}$, 
      $\scC\in C_{\sidx{} +1}\setminus C_{\sidx{}}$,
      and $\{\scA,\scB,\scC\}\subseteq \symdif{C_{\sidx{}+j-1}}{C_{\sidx{}+j}}$
      for each~$j \in \{1, \ldots, i -\sidx{}\}$.

      Towards constructing the solution~$\calC^*$, let~$q<i$
      be the smallest index such that committees~$C_q$
      to~$C_i$ form an alternating sequence of any kind. That
      is, in case alternating sequences of two types exist,
      we select the value of~$q$ that is related to the
      longer one (clearly, it cannot happen that the lengths
      of the two sequences are equal).

      If~$q=1$, then let $\calC^* = \switch(\calC, 1, \scA, \scC)$. Naturally,
      $\calC^*$ is a correct solution due to the features of the switch
      operator. Otherwise, we fix some value~$q>1$ and consider two cases
      depending on what is the type of the alternating sequence related to~$q$.

      Suppose $q$ defines a~$\{\scA,\scB\}$-alternating sequence. Hence, we have that
      one of the two holds:
      \begin{enumerate}
        \item (i.1)~$\scC \not\in C_{q-1}$,
        \item (i.2)~$\scC \in C_{q-1}$ and ($\scA\in C_{q-1}$ or~$\scB\in C_{q-1}$).
      \end{enumerate}
      In case~(i.1), we let $\calC^* = \switch(\calC, q, \scA, \scC)$.
      Since~$C_q$ is the first committee of the~$\{\scA,\scB\}$-alternating
      sequence, $C^*_q \ceq C_q \setminus \{\scA\} \cup \{\scC\}$. So, $\scC
      \notin \symdif{C_{q-1}}{C_q}$ but~$\scC \in \symdif{C_{q-1}}{C^*_q}$,
      which guarantees that~$|\symdif{C_{q-1}}{C^*_q}| \geq
      |\symdif{C_{q-1}}{C_q}|$, which in turn prove the correctness of~$\calC^*$.
      Now, consider~Case~(i.2), in which it holds that~$\scC\in C_{q-1}$.
      If~$\scA\in C_{q-1}$, then let $\calC^* = \switch(\calC, q, \scA, \scC)$.
      If, however, $\scA\not\in C_{q-1}$ and~$\scB\in C_{q-1}$, 
      then let $\calC^* = \switch(\calC, q, \scB, \scC)$. 
      Let~$c_0$ be the candidate that
      we switch with~$\scC$ in the just-described two ways of
      obtaining~$\calC^*$. 
      Then,
      the correctness of~$\calC^*$ follows from
      relabeling candidate~$\scC$ with~$c_0$ in the reasoning for Case~(i.1).
        
      The remaining case is that~$q>1$ defines a $\scC$-alternating sequence.
      Then one the following must hold:
      \begin{enumerate}
        \item (ii.1)~$\scC \in C_{q-1}$,
        \item (ii.2)~$\scA\not\in C_{q-1}$ or~$\scB\not\in C_{q-1}$.
      \end{enumerate}
      In Case~(ii.1), 
      we let $\calC^* = \switch(\calC, q, \scB, \scC)$. 
      In Case~(ii.2), 
      assuming that~$\scA\not\in C_{q-1}$, 
      we again let $\calC^* = \switch(\calC, q, \scA, \scC)$. 
      However, 
      if in Case~(ii.2) it holds that $\scB\not\in C_{q-1}$ and $\scA\in C_{q-1}$, 
      then we let $\calC^* = \switch(\calC, q, \scB, \scC)$. 
      The correctness of~$\calC^*$ in each of
      the aforementioned cases follows from arguments analogous to those for the
      case of~$\{\scA,\scB\}$-alternating sequence 

  \medskip\noindent
  Using~\Cref{tab:case-distinction} one can verify that the presented cases are
  exhaustive. In every case we obtain a correct solution~$\calC^* = (C_1^*,\dots,C_\tau^*)$
  by applying the switch operator. In particular, in every case, this operator
  is also applied on committee~$C_i$. Hence, in solution~$\calC^*$ it is always
  the case that either $\scA \in C^*_i$ or~$\scB \in C^*_i$. What follows, is
  that~$\calC^*$ is a solution where the smallest index~$i'$ with
  $|C_{i'}^*\setminus\{c,c'\}|=k$ either does not exist or is strictly larger
  than~$i$---each giving a contradiction to our initial assumption on~$i$.
  \begin{table}
    \centering%
    \begin{tabular}{lcccccc}
        & & \multicolumn{5}{c}{result of~$\{\scA, \scB\} \cap
    C_{i-1}$}\\\cmidrule(l{2em}r{2em}){3-7}
        & no~$C_{i-1}$ & $\emptyset$ & $\{\scA\}$ & $\{\scB\}$ &
        \multicolumn{2}{c}{$\{\scA,
  \scB\}$}\\\midrule
      $i = 1$ & Case~1 & \cellcolor{black!30} & \cellcolor{black!30} & \cellcolor{black!30} &
    \multicolumn{2}{c}{\cellcolor{black!30}}\\
      $i = 2$ & \cellcolor{black!30} & Case 2a &
      Case 2c & Case 2b & \multicolumn{2}{c}{Case 3}\\\addlinespace[.3em]
      & & & & & $\scC \in C_{i-1}$&
      $\scC \not\in C_{i-1}$\\\cmidrule{6-7}
      $i \geq 3$ & \cellcolor{black!30} & Case
      2a & Case 2c & Case 2b & Case 4 & Case 5 \\
    \end{tabular}
    \caption{An overview of the case distinction of in the proof
    of~\ref{lem:rmsntv-k-between-m-n} depending on the value of~$i$ and the
    elements of~$C_{i-1}$. The grayed out cells represent
  non-existent cases.\label{tab:case-distinction}}
  \end{table}
  \end{proof}
}

Having~\cref{lem:rmsntv-k-between-m-n}, we are now ready to
prove the correctness of~\cref{rrule:rmpv-m-larger-than-max-nk-tau}.

\begin{proof}[Proof of~\cref{rrule:rmpv-m-larger-than-max-nk-tau}'s correctness]
 Since~$m\geq k>n\tau+3$,
 there are two candidates~$\scA,\scB\in C$ nowhere approved.
 
 \RD{}
 Due to~\cref{lem:rmsntv-k-between-m-n}, 
 we know that there is a
 solution~$C_1,\dots,C_\tau$ such that~$C_i'\ceq C_i'\setminus\{\scA,\scB\}$ is
 of size at most~$k-1$.  
 Note that~$|\symdif{C_i'}{C_{i+1}'}|\geq
 |\symdif{C_i}{C_{i+1}}|-2\geq \ell-2$.
 
 \LD{}
 Given any solution~$C_1,\dots,C_\tau$, 
 adding~$\scA$ to~$C_{2i-1}$ and~$\scB$
 to~$C_{2i}$ gives a solution.
\end{proof}

\noindent
\cref{thm:pkntau} for \rmpv{} now follows:
If~$\ell>m$, then we return a trivial \no-instance. 
Throughout,
if~$k>m$, then we can set~$k$ to~$m$,
and if~$\ell<2$,
then we solve \rmpv{} in polynomial time~\citep{KRZ21}. 
Next, 
we exhaustively apply~\cref{rrule:rmpv-m-larger-than-max-nk-tau}.
Then, we have~$m\leq n\tau+3$,
or if not,
then we have~$k\leq n\tau+3$.
In the latter case,
after exhaustive application of~\cref{rrule:rmpv-m-larger-than-max-nk-tau-O}.
we then have that~$m\in O(n\cdot \tau^2)$.
Since we have polynomially-many applications each running in polynomial time, 
the theorem's part follows.

\toappendix{
\subsection{Fixed-Parameter Tractability of Weighted Variant}

\begin{proposition}%
 \label{prop:cmpvxtau}
 \Wcmpv{} is in \FPT{} regarding~$x+\tau$.
\end{proposition}

\begin{proof}
 Let~$I=((w^1,\dots,w^\tau),k,\ell,x)$ be an instance of \Wcmpv{}.
 It is not hard to see that we can obtain an equivalent instance~$\tilde{I}=((\tilde{w}^1,\dots,\tilde{w}^\tau),k,\ell,x)$
 such that~$\norm{\tilde{w}^t}{\infty}\leq x$ for all~$t\in\set{\tau}$.
 Now observe that each candidate~$i$ has a \emph{fingerprint} $(w_i^1,\dots,w_i^\tau)\in\set[0]{x}^\tau$,
 and that there are at most~$(x+1)^\tau$ pairwise different fingerprints.
 We have the following:
 
 \begin{claim}%
 \label{clm1}
  If there are more than~$x$ many candidates with the same fingerprint,
  then we can delete one of them.
 \end{claim}
 
 \begin{proof}[Proof of the claim.]
  Let~$(C_1,\dots,C_\tau)$ be a solution such that there is at least one~$i\in\set{\tau}$
  such that~$C_i$ contains~$x'>x$ candidates~$c_1',\dots,c_{x'}'$ with the same fingerprint.
  Moreover,
  we assume that candidates with the same fingerprint are lexicographic smallest
  (this does not change the score, yet possibly decrease the symmetric difference between consecutive committees).
  Then delete~$c_{x'}'$ from all~$C_i$ to obtain the sequence~$(C_1',\dots,C_\tau')$.
  Note that any size of a committee and any symmetric difference does not increase,
  and the score of each committee is still at least~$x$.
  Thus,
  $(C_1',\dots,C_\tau')$ is in fact a solution.
 \end{proof}
 
 Due to the claim,
 we can assume that~$m\leq x(x+1)^\tau$.
 In total,
 we constructed an instance of size~$(x+1)^{O(\tau)}$,
 and hence,
 \Wcmpv{} is in \FPT{} when parameterized by~$x+\tau$.
\end{proof}

}

\cref{thm:pkntau,thm:PKwrtmtau} emphasize the role of a short time-horizon in
our problems as they both deal with instances with few stages.
In passing, 
we point out that the above theorems also imply that the corresponding
parameterized problems are in FPT. 
This is due to a well-known fact 
that 
a
kernelization implies fixed-parameter
tractability
for the same parameter.

\section[Conclusion and Discussion]{Conclusion and Discussion}
\label{sec:conclusion}
\appendixsection{sec:conclusion}

\toappendix{
\subsection{Number of candidates and agents combined}
\label{sec:mn}

When parameterized by~$m+n$ polynomial problem kernelization turns out to be unlikely,
which may surprise, since one parameterizes by all dimensions of an input except for the time aspect (number of stages).

\begin{theorem}
 \label{thm:nopkmn}
 Each of \cmpv{} and \rmpv{} admits no problem kernel of size polynomial in~$m+n$ unless~$\NPincoNPslashpoly$.
\end{theorem}

\begin{proof}
To prove~\cref{thm:nopkmn},
we are going to prove that both,~\cmpv{} and \rmpv{}, when parameterized by~$m+n$ are AND-compositional.
That is,
there is an algorithm taking~$p$ instance~$I_1,\dots,I_p$,
each with the same number~$n$ of agents, 
$m$ of candidates, 
and~$\tau$ of profiles,
and constructs in time polynomial in~$\sum_{i=1}^p |I_i|$ an instance~$I$ such that the number of agents and candidates is in~$(m+n)^{O(1)}$,
and that~$I$ is a \yes-instance if and only if each of~$I_1,\ldots,I_p$ is \yes-instance.
When a parameterized problem admits an AND-composition,
then it admits no polynomial kernel unless~$\NPincoNPslashpoly$.

\begin{construction}[\cmpv{}]
 \label{constr:compocmpv}
 Consider~$p$ instances~$I_1,\dots,I_p$ of~\cmpv{} such that~$k,x\in \N$ and~$\ell=1$, 
 each with~$n$ agents~$A^i=\{a_1^i,\ldots,a_n^i\}$,
 $m$ candidates~$C^i=\{c_1^i,\ldots,c_m^i\}$,
 and~$\tau$ voting profiles.
 Add the agent set~$A=\{a_1,\ldots,a_n\}\cup\{\wmod{a}_1,\ldots,\wmod{a}_{n}\}$ and candidate set~$C=\{c_1,\ldots,c_n\}\cup\{z\}$.
 Arrange the~$p$ voting profiles consecutively,
 where we identify each agent~$a_j^i$ with~$a_j$ and each candidate~$c_j^i$ with~$c_j$.
 Moreover,
 each agent~$\wmod{a}_j$ approves~$z$.
 Next,
 between the last voting profile and the first voting profile of two consecutive instances,
 add~$2k$ profiles where all agents approve~$z$ (we call them \emph{transfer profiles}).
 Set~$x'=x+n$,
 $k'=k+1$, and
 $\ell'=1$.
 Note that~$\tau'= p\tau+2k(p-1)$.
 \cqed{}
\end{construction}
Now we prove the correctness of~\cref{constr:compocmpv}. By construction, we
have the following.

\begin{observation}
 \label{obs:zisinall}
 Let~$I'$ from~\cref{constr:compocmpv} be a \yes-instance of~\cmpv{}.
 Then, it holds true for every solution~$(C_1,\ldots,C_{\tau'})$ that~$z\in C_t$ in all~$t\in\set{\tau'}$.
\end{observation}

\begin{proof}
 Let~$\calC=(C_1,\ldots,C_{\tau'})$ be a solution to~$I'$.
 Assume there is~$i\in\set{\tau'}$ such that~$z\not\in C_i$.
 Then,~$\score_i(C_i)\leq |A|-n = n < n+x = x'$,\
 contradicting that~$\calC$ is a solution.
\lqed
\end{proof}

\begin{lemma}
 Let~$I'$ be the instance obtained from~\cref{constr:compocmpv} given~$I_1,\dots,I_p$.
 Then~$I'$ is a \yes-instance if and only if each instance of~$I_1,\dots,I_p$ a \yes-instance.
\end{lemma}

\begin{proof}
 \RD{}
 Let~$\calC=(\calC_1,\calC_{1,2},\calC_2,\dots,\calC_p)$ be a solution for~$I'$,
 where~$\calC_i$ is sequence of committee for the voting profiles obtained from instance~$I_i$,
 and~$\calC_{i,i+1}$ is a sequence of committees for the transfer voting profiles between the consecutive instances~$I_i$ and~$I_{i+1}$.
 We claim that~$\calC_i$ without~$z$ forms a solution to instance~$I_i$ for every~$i\in\set{p}$.
 Due to \cref{obs:zisinall},
 we have that~$z$ is contained in each committee in~$\calC_i$.
 Moreover,
 each committee admits a score of at least~$x'=x+n$,
 and hence as~$z$ contributes~$n$ to the score,
 each committee admits a score of at least~$x$ in the respective voting profile in instance~$I_i$.
 
 \LD{}
 Let each of $I_1,\dots,I_p$ be a \yes-instance,
 and let~$\calC_q=(C_q^1,\dots,C_q^\tau)$ be a solution to~$I_q$ for each~$q\in\set{p}$.
 Note that~$\wmod{\calC}_q=(\wmod{C}^1_q,\dots,\wmod{C}_q^\tau)$ where $\wmod{C}_q^t$ contains the set~$\{c_j\mid c_j^q\in C_q^t\}$ and~$z$
 is a sequence of committees that is a partial solution to~$I'$ on the voting profiles corresponding to~$I_q$.
 It remains to construct a sequence~$\calC_{q,q+1}=(C_{q,q+1}^1,\ldots,C_{q,q+1}^{2k})$ for each~$q\in\set{p-1}$.
 Let~$\wmod{C}_q^\tau=\{z\}\cup \{b_1,\ldots,b_r\}$ and let~$\wmod{C}_{q+1}^1=\{z\}\cup \{d_1,\ldots,d_s\}$.
 Initially, 
 set~$C_{q,q+1}^t=\{z\}$ for all~$t\in\set{2k}$.
 For~$t\leq r$,
 set~$C_{q,q+1}^t=\wmod{C}_q^\tau\setminus\bigcup_{i=1}^t{b_i}$.
 For~$t\geq 2k-s+1$,
 and~$C_{q,q+1}^t=\wmod{C}_{q+1}^1\setminus\bigcup_{i=1}^{2k-t+1}{d_i}$.
 Clearly,
 $|C_{q,q+1}^t|\leq k'$ and since each contains~$z$,
 each admits a score of~$2n$.
 Finally,
 by construction we have that~$|\symdif{C_{q,q+1}^t}{C_{q,q+1}^{t+1}}|\leq 1$ for every~$t\in\set{2k-1}$,
 $|\symdif{C_{q}^\tau}{C_{q,q+1}^{1}}|\leq 1$,
 and~$|\symdif{C_{q+1}^1}{C_{q,q+1}^{2k}}|\leq 1$.
 Altogether, it holds that
 $\calC\ceq (\wmod{\calC}_1,\calC_{1,2},\wmod{\calC}_2,\dots,\wmod{\calC}_p)$ forms a solution to~$I'$.
\lqed
\end{proof}

For~\rmpv{},
the construction is similar to~\cref{constr:compocmpv} in the sense of using transition profiles; 
However,
this time, we only need three.

\begin{construction}
 \label{constr:compormpv}
 Consider~$p$ instances~$I_1,\dots,I_p$ of~\rmpv{} such that $k,x\in \N$ and~$\ell=2k$, 
 each with~$n$ agents~$A^i=\{a_1^i,\ldots,a_n^i\}$,
 $m=\ell$ candidates~$C^i=\{c_1^i,\ldots,c_m^i\}$,
 and~$\tau$ voting profiles.
 Add the agent set~$A=\{a_1,\ldots,a_n\}\cup\{\wmod{a}_j^i\mid 0\leq i\leq \ell, 1\leq j\leq n\}$ 
 and candidate set~$C=\{c_1,\ldots,c_n\}\cup\{z,y_1,\dots,y_\ell\}$.
 Arrange the~$p$ voting profiles consecutively,
 where we identify each agent~$a_j^i$ with~$a_j$ and each candidate~$c_j^i$ with~$c_j$.
 Moreover,
 each agent~$\wmod{a}_j^0$ approves~$z$,
 and each~$\wmod{a}_j^i$ approves~$y_i$.
 Next,
 between the last voting profile and the first voting profile of two consecutive instances,
 add~one profile where all agents approve~$z$ (we call it again \emph{transfer profile}).
 Set~$x'=x+n\cdot(\ell+1)$,
 $k'=k+\ell+1$, and
 $\ell'=\ell$.
 Note that~$\tau'= p\tau+(p-1)$.
 \cqed{}
\end{construction}

Again, 
similar to the case of~\cmpv{},
we observe that~$z$ must be contained in each committee when we are facing a solution.

\begin{observation}
 \label{obs:zisinallrmpv}
 Let~$I'$ from~\cref{constr:compormpv} be a \yes-instance of~\rmpv{}.
 Then it holds true for every solution~$(C_1,\ldots,C_{\tau'})$ that~$z\in C_t$ in all~$t\in\set{\tau'}$ and~$\{y_1,\ldots,y_\ell\}\in C_t$ if~$C_t$ is no transition profile.\end{observation}

\begin{lemma}
 Let~$I'$ be the instance obtained from~\cref{constr:compormpv} given~$I_1,\dots,I_p$.
 Then~$I'$ is a \yes-instance if and only if each instance of~$I_1,\dots,I_p$ a \yes-instance.
\end{lemma}

\begin{proof}
 \RD{}
 Let~$\calC=(\calC_1,\calC_{1,2},\calC_2,\dots,\calC_p)$ be a solution for~$I'$,
 where~$\calC_i$ is sequence of committee for the voting profiles obtained from instance~$I_i$,
 and~$\calC_{i,i+1}$ is a sequence of committees for the transfer voting profiles between the consecutive instances~$I_i$ and~$I_{i+1}$.
 We claim that~$\calC_i$ without~$z$ forms a solution to instance~$I_i$ for every~$i\in\set{p}$.
 Due to \cref{obs:zisinallrmpv},
 we have that~$z$ and~$\{y_1,\dots,y_\ell\}$ are contained in each committee in~$\calC_i$.
 Moreover,
 each committee admits a score of at least~$x'=x+n\cdot(\ell+1)$,
 and hence as each of~$z,y_1,\dots,y_\ell$ contributes~$n$ to the score,
 each committee admits a score of at least~$x$ in the respective voting profile in instance~$I_i$.
 
 \LD{}
 Let each of $I_1,\dots,I_p$ be a \yes-instance,
 and let~$\calC_q=(C_q^1,\dots,C_q^\tau)$ be a solution to~$I_q$ for each~$q\in\set{p}$.
 Note that~$\wmod{\calC}_q=(\wmod{C}^1_q,\dots,\wmod{C}_q^\tau)$ where $\wmod{C}_q^t$ contains the set~$\{c_j\mid c_j^q\in C_q^t\}$ and~$z$
 is a sequence of committees that is a partial solution to~$I'$ on the voting profiles corresponding to~$I_q$.
 It remains to construct the transition profiles~$C_{q,q+1}$ for each~$q\in\set{p-1}$.
 Let~$\wmod{C}_q^\tau=\{z\}\uplus \{y_1,\dots,y_\ell\} \uplus \{b_1,\ldots,b_r\}$ and let~$\wmod{C}_{q+1}^1=\{z\}\uplus \{y_1,\dots,y_\ell\} \uplus \{d_1,\ldots,d_s\}$.
 Set~$C_{q,q+1}=\{z\}$.
 Clearly,
 $|C_{q,q+1}|\leq k'$,
 Moreover,
 since in the transition profiles all agents approve~$z$,
 committee~$C_{q,q+1}$ has score~$A'\geq x'$.
 Finally,
 observe that~$\{y_1,\dots,y_\ell\}$ is a subset of each of~$\symdif{\wmod{C}_q^\tau}{C_{q,q+1}}$ and of~$\symdif{C_{q,q+1}}{\wmod{C}_{q+1}^1}$,
 and hence each symmetric difference is of size at least~$\ell$.
 Altogether,
 $\calC\ceq (\wmod{\calC}_1,C_{1,2},\wmod{\calC}_2,\dots,\wmod{\calC}_p)$ forms a solution to~$I'$.
\lqed
\end{proof}
The above lemma concludes the proof of all claims of~\cref{thm:nopkmn}.
\end{proof}
}

While our multivariate analysis 
revealed several intractability results,
we emphasize
that
we also identified quite practically relevant 
tractable cases.
In natural applications, such as electing committees serving for only few
days/events, the number~$\tau$ of stages is usually small. Even more, since
planning too far in advance usually increases uncertainty. Furthermore, usually
either the number~$n$ of voters or the number~$m$ of candidates is expected to
be small (as suggested by quite a significant number of small-sized election in
the preflib database). Additionally, in many elections the committee size~$k$ is
not very large.  Hence, our results for~$\tau$ and for~$k$ (polynomial-time
solvability for constant values) as well as for~$\tau +n$, for~$\tau + m$, or
for $m$~alone (fixed-parameter tractability) form a very
positive message.

Our results also underline the importance of the time aspect for preprocessing.
For both~\cmpv{} and \rmpv{}, we show that efficient data reduction to size
polynomial in the number of agents and the number of candidates is
unlikely (see \cref{app:sec:conclusion}).
Yet, combining any of the
two parameters with the number of stages (see~\cref{thm:pkntau,thm:PKwrtmtau})
allows for efficient (polynomial) data reduction.

Moreover, the tractability result for~$\tau$ seems generally insightful for the
multistage community where problems usually remain computationally hard even a
for constant number of stages.
Additionally,
the revolutionary multistage model may be relevant on its own
and 
may
pave the way for studying more new
models where consecutive changes are both lower- and upper-bounded.
This is already underlined by the study of~\citet{KRZ21}
on several further problems like matching or $s$-$t$ path
taking up our revolutionary setup.

\paragraph{Representation.}
Although it is naturally justified to start with the simplest meaningful model variant,
our focus on SNTV might look restrictive.
We stress that most results transfer easily to general Approval profiles (see\ \Wcmpv):
Replace each voter approving multiple candidates by multiple voters, each approving one candidate.
Also basic scoring rules can be modeled: Create for each candidate~$c$ that receives score~$s(c)$ exactly~$s(c)$ votes for~$c$.
The main drawback is that now the parameter number~$n$ of voters corresponds to the total number
of approvals or the total score sum, respectively; yet, most positive results still hold.
It remains open whether a direct modeling of more complex preferences instead of blowing up in the number of voters
can avoid blowing up the running time as well.
Recently, 
it was shown 
that many of our results
transfer to more complex voting rules.\footnote{Burak Arinalp. Multistage committee elections: Beyond plurality voting, March, 2021. Bachelor thesis. TU Berlin.}

{
\paragraph{Deeper comparison of our models.}
As opposed to the single-stage case, both
conservative and revolutionary committee election %
over multiple stages are \NP-complete,
even for a constant number of agents.
From a parameterized algorithmic point of view,
computing a revolutionary
committee is easier than a conservative one:
When asking for committees to change for all but constantly many candidates,
\rmpv{} is polynomial time solvable,
yet when asking for committees to change for only constant many candidates,
\cmpv{} remains \NP-hard.
Finally,
we wonder if
\rmpv{}
parameterized by~$k+\tau$ or~$\ell+\tau$ admits a polynomial problem kernel,
and whether \rmpv{}
parameterized by~$x+\tau$ is in~\FPT{}.

\paragraph{Future work.}
Possible future directions may include studying approximate or randomized algorithms
for~\cmpv{} and~\rmpv{} as well as experimentally testing our results.
Moreover,
further concepts and problems (e.g.,\ bribery and manipulation) from computational social choice may be studied in the (conservative and revolutionary) multistage model.
Note that, for instance, 2-Approval (each agent approves up to two candidates)
in the conservative multistage setup is already \NP-hard for one agent
(similar to the proof of~\cref{thm:bothnphard}(i)).
As a concrete future work,
we want to prove 
that \cmpv{} is \W{1}-hard when parameterized by~$k+n$
(inspired by a proof of~\citet{FluschnikNRZ22}).

\paragraph{Offline versus online.}
Importantly, our model is applicable for offline scenarios, in which preferences
are collected in advance (e.g.,\ by
social media polls, Internet
profiling, customer targeting).
However,
online
scenarios also offer an interesting research direction, yet requiring significant
changes in our original models. To observe this, consider an online scenario of
selecting
two committees such that in the first profile all committees are scoring equally high
and in the second profile there is exactly one such a committee.
In the worst case, every algorithm returns a solution requiring exchanging
all candidates between the two selected committees; thus, no reasonable
guarantee concerning the number of changes is achievable. To avoid such trivial
cases, when studying the online setting, one needs to carry out significant
model modifications. One way to proceed (following
the
multistage literature~\citep{BampisEST19,BampisEM18,GuptaTW14}) is to
introduce a goal function and consider
the quality of the selected committees
and the symmetric difference as soft constraints.
Another way is
to
restrict the differences of consecutive profiles (analogously to \citet{PP13}).
Such a correlation between consecutive profiles (greatly restricting the model)
provides enough information %
to achieve some guarantees on the solution.

\section* {Acknowledgements}
\begin{wrapfigure}{r}{.12\textwidth}
\begin{center}
\includegraphics[width=1.85cm]{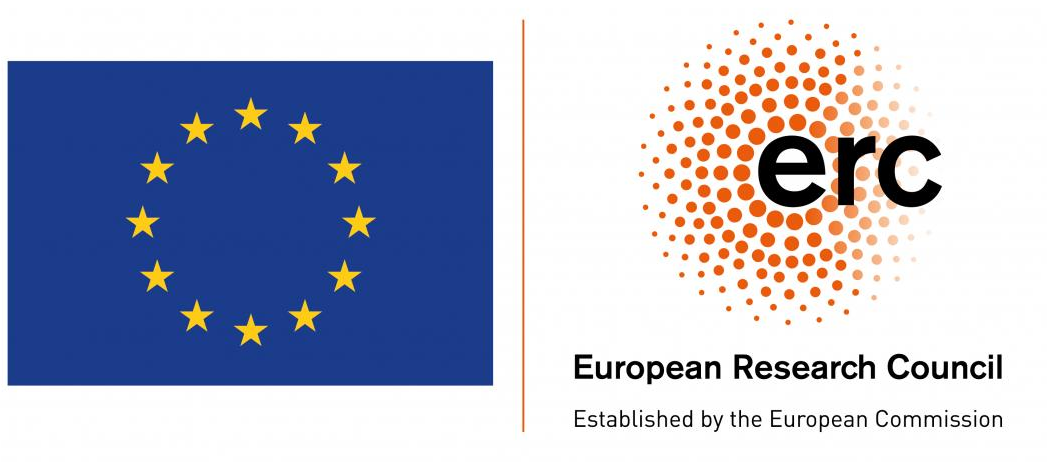}
\end{center}
\end{wrapfigure}
We thank the IJCAI'22 reviewers for their helpful comments.
This work was started when all authors were with TU~Berlin.
TF was supported by the DFG, projects TORE (NI 369/18),
MATE (NI 369/17),
and AFFA (BR 5207/1).
AK was supported by the DFG, project AFFA (BR 5207/1 and NI 369/15),
and by the European Research Council (ERC).
This project has received funding from the European Research Council (ERC)
under the European Union’s Horizon 2020 research and innovation programme (grant
agreement No 101002854).

\clearpage
\newpage

\bibliographystyle{plainnat}
\bibliography{msce-bib}

 \newpage
 \clearpage
 \appendix
 \section*{\section*{Appendix}}
 \appendixProofText
\end{document}